\documentclass[11pt]{article}

\usepackage[utf8]{inputenc}
\usepackage{fullpage}
\usepackage{amsmath,amsthm,amsfonts,amssymb}
\usepackage[textwidth=0.8in,colorinlistoftodos]{todonotes}
\usepackage[capitalize,nameinlink]{cleveref}
\usepackage{mathtools}
\usepackage{braket}
\usepackage{thmtools,thm-restate}
\usepackage{enumitem}
\numberwithin{equation}{section}

\usepackage{microtype}
\usepackage{mleftright}
\usepackage[mathscr]{euscript}

\newtheorem{theorem}{Theorem}[section]
\newtheorem{lemma}[theorem]{Lemma}
\newtheorem{definition}[theorem]{Definition}

\newtheorem{claim}[theorem]{Claim}

\newtheorem{construction}[theorem]{Construction}
\newtheorem{introtheorem}{Theorem}

\newcommand{\eps}{\varepsilon}
\newcommand{\bitset}{\{0,1\}}
\newcommand{\tritset}{\{0,1,\bot\}}
\newcommand{\N}{\mathbb N}

\newcommand{\polylog}{\mathrm{polylog}}
\newcommand{\supp}{\mathrm{supp}}

\newcommand{\decrad}{\delta}
\newcommand{\locality}{\ell}
\newcommand{\dist}{\Delta}
\newcommand{\adist}{\bar{\Delta}}

\renewcommand{\deg}[2]{\mathrm{deg}_{#1}(#2)}

\newcommand{\sunflower}{\mathcal{S}}
\newcommand{\petals}{\mathcal{P}}
\newcommand{\good}{\mathcal{G}_i}

\newcommand{\binomparam}{n^{-1/{2\locality'^2}}}
\newcommand{\bp}{1/n^{1/2\locality^2}}
\newcommand{\sbparam}{n^{1-1/2{\locality'^2}}}

\newcommand{\daisy}[1]{$#1$-daisy}
\newcommand{\simpledaisy}{\daisy{1}}

\newcommand{\textdef}[1]{\textnormal{\textsf{#1}}}

\title{A Lower Bound for Relaxed Locally Decodable Codes}

\author{Tom Gur \\ University of Warwick\\ \texttt{tom.gur@warwick.ac.uk} \and
Oded Lachish \\ Birkbeck, University of London \\ \texttt{oded@dcs.bbk.ac.uk}}
\date{}

\begin{document}
\maketitle

\begin{abstract}
A locally decodable code (LDC) $C \colon \bitset^k \to \bitset^n$ is an error correcting code wherein individual bits of the message can be recovered by only querying a few bits of a noisy codeword. LDCs found a myriad of applications both in theory and in practice, ranging from probabilistically checkable proofs to distributed storage. However, despite nearly two decades of extensive study, the best known constructions of $O(1)$-query LDCs have super-polynomial blocklength.

The notion of relaxed LDCs is a natural relaxation of LDCs, which aims to bypass the foregoing barrier by requiring local decoding of nearly all individual message bits, yet allowing decoding failure (but not error) on the rest. State of the art constructions of $O(1)$-query relaxed LDCs achieve blocklength $n = O\left(k^{1+ \gamma}\right)$ for an arbitrarily small constant $\gamma$.

We prove a lower bound which shows that $O(1)$-query relaxed LDCs cannot achieve blocklength $n = k^{1+ o(1)}$. This resolves an open problem raised by Goldreich in 2004.
\end{abstract}

\newpage
\tableofcontents

\newpage
\section{Introduction}
Locally decodable codes (LDC) are fundamental objects in coding theory. Loosely speaking, an LDC is an error correcting code with a robust local-to-global structure which admits a randomised algorithm that can recover individual message bits by probing a minuscule portion of a noisy codeword. Thus, rather than reading the entire codeword to decode the entire message, an LDC allows for reading a small number of locations to decode a single bit of the message.

More precisely, we consider codes $C \colon \bitset^k \to \bitset^n$ with linear distance, and say that a code $C$ is an LDC if there exists a randomised algorithm, called a \textsf{local decoder}, that is given a location $i \in[k]$ and query access to an input $w\in\bitset^n$  such that if $w$ is sufficiently close (typically within distance that is proportional to the distance of the code) to a valid codeword $C(x)$, the decoder outputs $x_i$ with high probability. The maximal number of queries that the decoder makes is called the \textsf{locality} of the code.

Since the systematic study of LDCs was initiated in the seminal work of Katz and Trevisan \cite{KT00}, these codes received much attention and made a profound impact on cryptography, complexity theory, program checking, data structures, quantum information, pseudorandomness, and other areas in theoretical computer science (see surveys \cite{Trevisan04,Yekhanin12,KS17} and references therein), as well as led to significant practical applications in distributed storage \cite{Huang2012}.

Unfortunately, despite the success and attention that LDCs gained in the last two decades, there remains a chasm between the best known upper bound and lower bounds on LDCs. Specifically, the best general lower bounds that are currently known (cf. \cite{KW04,Woodruff12},\footnote{For specific regimes of parameters, some improvements are known, e.g., for $3$-local LDCs, the blocklength must be $\tilde{\Omega}\left(k^{1+\frac{1}{\lceil q/2\rceil - 1}}\right)$  \cite{KW04}.} building on \cite{KT00}), show that any $\locality$-local LDC must have blocklength 
\begin{equation*}
    n = \Omega\left(k^{1+\frac{1}{\locality - 1}}\right), 
\end{equation*}
where throughout, $k$ is the dimension of the code. In stark contrast, the state-of-the-art construction of $O(1)$-local LDCs has a \emph{super-polynomial} blocklength (cf. \cite{Efremenko12}, building on \cite{Yekhanin08}).

The foregoing barrier has led to the study of \emph{relaxed locally decidable codes}, in short ``relaxed LDCs'', which were introduced in the highly influential work of Ben-Sasson, Goldreich, Harsha, Sudan, and Vadhan \cite{BGHSV04}. In a recent line of research \cite{GGK15,GG16a, GG18, GR18,GRR18,Blocki18} relaxed LDCs and their variants (such as relaxed locally correctable codes) have been studied and used to obtain applications to property testing \cite{CG18}, data structures \cite{CGW09}, and probabilistic proof systems \cite{GGK15,GG16a,GR18}.

Loosely speaking, this relaxation of LDCs allows the local decoder to declare ``decoding failure'' on a small fraction of the indices, yet crucially, still avoid errors. More accurately, a \textsf{relaxed LDC} $C \colon \bitset^k \to \bitset^n$ is a code (with linear distance) for which there exists a \textsf{decoding radius} $\decrad$ (typically proportional to the relative distance of the code) and a probabilistic algorithm, called the \textsf{relaxed local decoder} that receives an index $i \in [k]$ and oracle access to a string $w \in\bitset^n$ that is $\decrad$-close to a codeword $C(x)$. The relaxed local decoder is allowed to make a small number of queries to $w$ (typically $O(1)$ queries) and is required to satisfy the following conditions:
\begin{enumerate}
\item \emph{Completeness:} If the input is a valid codeword (i.e., $w = C(x)$), the relaxed local decoder must always output $x_i$.
\item \emph{Relaxed decoding:} Otherwise, with high probability, the decoder must either output $x_i$ or a special ``reject'' symbol $\bot$ (indicating the decoder detected an error and is unable to decode).
\end{enumerate}

As observed in \cite{BGHSV04}, the foregoing two conditions suffice for obtaining a third condition, which guarantees that the relaxed local decoder may only reject (i.e., output $\bot$) on an arbitrarily small fraction of the coordinates. (See \cref{sec:preliminaries} for a formal definition of relaxed LDCs, covering all three conditions.)

This seemingly modest relaxation turns out to allow the usage of extremely powerful tools from the theory of probabilistically checkable proofs (PCPs). Relying on the notion of PCPs of proximity, which they also introduced and constructed, Ben-Sasson et al. \cite{BGHSV04} constructed a relaxed LDC with \emph{nearly-linear length}. More precisely, they showed that for every constant $\gamma>0$ there exists an $O(1)$-local relaxed LDC $C \colon \bitset^k \to \bitset^n$ with nearly-linear blocklength $n = k^{1+\gamma}$. We remark that 15 years later, there is no known construction of relaxed LDCs that improves on \cite{BGHSV04}.

While the aforementioned relaxed LDCs have blocklength that is nearly exponentially shorter than that of any known non-relaxed LDC, they do \emph{not} break the currently known lower bound on non-relaxed LDCs (cf. \cite{KT00}). This led Goldreich \cite{Goldreich04} to raise the following open problem:

\begin{center}
    Do there exist $O(1)$-local relaxed LDCs with blocklength $n = k^{1+o(1)}$?
\end{center}

\subsection{Our results}
Our main contribution resolves the foregoing open problem by providing a strong negative answer. Namely, we prove the following theorem, which shows that $O(1)$-local relaxed LDCs cannot achieve blocklength $n = k^{1+ o(1)}$.

\begin{introtheorem}
    \label{thm:main}
    For any $\locality, \decrad \in \N$, there exists a constant $\alpha = \alpha(\locality,\decrad)$ such that every $\locality$-local relaxed LDC $C \colon \bitset^k \to \bitset^n$ with decoding radius $\decrad$ satisfies $n = \Omega\left(k^{1+ \alpha}\right)$.
\end{introtheorem}

To the best of our knowledge, this is the first non-trivial lower bound that was shown for \emph{relaxed} LDCs. We remark that \cref{thm:main} directly extends to the setting of linear \emph{relaxed locally correctable codes}, recently introduced in \cite{GRR18}.

\paragraph{On adaptivity}
For relaxed LDCs with constant decoding radius and non-adaptive decoders (i.e., where each query is made independently of the answers to previous queries), the parameter $\alpha$ in our lower bound of $n = \Omega\left(k^{1+ \alpha}\right)$ takes the form of $\alpha = \frac{1}{O(\ell^2)}$. This dependency is quite close to that in the upper bound of Ben-Sasson et al. \cite{BGHSV04}, who showed that for any $\locality \in \N$ there exists an $\locality$-local non-adaptive relaxed LDC $C \colon \bitset^k \to \bitset^n$ with blocklength $n = k^{1 + \Theta(1/\sqrt{\locality})}$.

\paragraph{On perfect completeness.} We remark that our techniques generalise to the setting of relaxed LDCs in which the completeness condition only requires that perfectly valid codewords can be locally decoded with probability $2/3$, rather than probability $1$ (i.e., perfect completeness). For simplicity sake, we focus on decoders with perfect completeness.

\subsection{Related works}
\label{sec:related-works}
There is an extensive literature that is concerned with lower bounds on (non-relaxed) locally decodable codes in various regimes (see, e.g., \cite{KT00,Deshpande02,GKST02,Obata02,KW04,WW05,Woodruff12}), as well as for the closely related notion of locally correctable codes (see, e.g., \cite{Barak2011,BDSS11,BGT16,DSW17}), in which the goal is to correct a bit of the codeword rather than a bit of the message. We stress that none of the aforementioned bounds apply for \emph{relaxed LDC} (see discussion in \cref{sec:challange}).

Another related notion is that of \emph{locally testable codes} \cite{GS06}, which are, loosely speaking, codes for which there exists a probabilistic algorithm that accepts valid codewords, and rejects inputs that are ``far'' in Hamming distance from any codeword, while only probing a small fraction of the input. Much stronger upper bounds are known for locally testable codes than for LDCs, and in particular, there exists $O(1)$-local LTCs with blocklength $n = k \cdot \polylog(k)$ \cite{GS06} (see also \cite{Meir09,Vid13}). It is also known that LDCs do \emph{not} imply locally testable codes and vice versa \cite{KV10}.

Locally testable codes can be viewed as a special case of \emph{property testing} (see recent book \cite{Goldreich17} and references therein), which deals with algorithms that distinguish whether an input belongs to a set $S$ or is ``far'' from any input in $S$. As is the case with locally testable codes, LDCs and property testing are very distinct notions. In particular, whereas a local decoder is a \emph{local computation} algorithm that operates under the guarantee that the input is \emph{close to a codeword}, a property tester is an \emph{approximate decision} algorithm that distinguishes between exact membership in a set, and being \emph{far} from the set. Interestingly, despite these fundamental differences, we are still able to rely on techniques from \cite{FLV15} that were used in the context of property testing (see \cref{sec:generalised-sunflower}). 

\subsection{Organisation}
The rest of the paper is organised as follows. In \cref{sec:techniques} we provide a high-level overview of our techniques. In \cref{sec:preliminaries} we cover the necessary preliminaries. Finally, in \cref{sec:proof} we prove \cref{thm:main}.

\section{Techniques}
\label{sec:techniques}
In this section, we provide an overview of the proof of the lower bound in \cref{thm:main}. We begin in \cref{sec:challange} by articulating the challenge in proving a lower bound on relaxed LDC and discuss why current techniques for non-relaxed LDCs lower bounds are inherently incompatible with the setting of relaxed LDCs.

In \cref{sec:approach}, we present our high-level strategy for obtaining the lower bound, which is centred around using the relaxed local decoder to obtain a ``global decoder''; that is, a probabilistic algorithm that decodes the \emph{entire} message of a \emph{perfectly valid} codeword. In \cref{sec:first-step}, we discuss a naive attempt towards constructing such a global decoder, and articulate two main technical challenges that arise.

In \cref{sec:generalised-sunflower}, we address the first challenge by arguing that the local views of relaxed local decoders can be assumed, without loss of generality, to satisfy a structure that can be thought of as a relaxation of combinatorial sunflowers. In \cref{sec:global}, we address the remaining challenge and present the construction of our global decoder. Finally, in \cref{sec:tech-analysis}, we discuss the analysis of the global decoder and how it implies the desired lower bound.

\subsection{The challenge}
\label{sec:challange}
As we mentioned in \cref{sec:related-works}, the coding theory literature has a large body of works that prove lower bounds on \emph{non-relaxed} LDCs. It is tempting to try and apply the methodology used in these works to our setting of \emph{relaxed} LDCs. 

The caveat, however, is that essentially all LDC lower bound techniques in the literature rely on the \emph{smoothness} property of LDCs (cf. \cite[Theorem 1]{KT00}). Loosely speaking, a decoder is said to be smooth if the distribution of queries that it makes is well-spread; that is, no coordinate is being queried with high probability by the decoder. The smoothness of LDCs provides structural insight regarding local decoders, which lie at the heart of these techniques.

In stark contrast, relaxed LDCs are \emph{not} necessarily smooth. In fact, all known constructions of non-trivial relaxed LDCs (i.e., which achieve parameters that are better than known for non-relaxed LDC) are highly non-smooth, in the following sense: for each message index $i \in [k]$, a significant fraction of the queries that the relaxed local decoder makes are concentrated on a small number of coordinates.

As observed in \cite{BGHSV04}, relaxed LDCs can be made to satisfy a weaker condition, known as \emph{average smoothness}, which states that the decoder makes nearly uniform queries on \emph{average}, taken over all indices $i \in [k]$ to be decoded (however, for any particular $i \in [k]$, the queries of $D$ given decoding index $i$ may be highly concentrated). Unfortunately, the average smoothness condition is a much weaker requirement than smoothness (e.g., see discussion in \cite[Section 4.2.1]{BGHSV04}), and it is highly unclear whether it can be used to imply relaxed LDC lower bounds.

Instead, to show a lower bound on relaxed LDCs we use a new methodology that does \emph{not} rely on smoothness at all to argue about the structure of the relaxed local decoder. The approach that we take, which we discuss next in \cref{sec:approach}, strongly relies on an observation that the structure of the local views that relaxed decoders make can be essentially captured by a relaxation of the notion of sunflowers, to which we refer as \emph{daisies} and discuss in \cref{sec:generalised-sunflower}.

\subsection{High-level approach}
\label{sec:approach}
Recall that our goal is to show that every $O(1)$-local relaxed LDC $C \colon \bitset^k \to \bitset^n$ with decoding radius $\decrad = O(1)$ satisfies $n = \Omega\left(k^{1+ \alpha}\right)$, for some constant $\alpha>0$ that depends on the locality parameter and decoding radius. 

Let $C$ be such $\locality$-local relaxed LDC for $\locality = O(1)$, and let $D$ be its corresponding relaxed local decoder. For clarity of exposition, throughout the techniques section we make the simplifying assumptions that $D$ has the following properties:
\begin{enumerate}[noitemsep]
    \item \emph{non-adaptive queries:} each query is made independently of the answers to previous queries,
    \item \emph{reduced error probability:} the decoder errs with probability at most $O(1/\locality^2)$,
    \item \emph{logarithmic randomness complexity:} the decoder uses $\log(n) + O(1)$ bits to generate its queries.
\end{enumerate} 
In the actual proof, we obtain these properties by adapting standard transformations to our setting, at the cost of deterioration in part of the parameters;\footnote{The loss in parameters due to the last two transformations is minor. The adaptive to non-adaptive transformation, on the other hand, increases the (constant) query complexity by an exponential factor.} see \cref{sec:preprocessing} for details.

Our strategy for proving \cref{thm:main} is to rely on the relaxed local decoder $D$ to construct a \emph{sample-based}, \emph{global} decoder $G$ for the code $C$. By sample based, we mean that the decoder $G$ queries each coordinate independently with a certain probability $p$. By global, we mean that $G$ decodes the \emph{entire} message of a \emph{perfectly valid} codeword. We stress that for our argument, it suffices for the global decoder to only work under the promise that the input codeword is not corrupted.

Our goal is to show that the global decoder $G$ would successfully decode the entire message, when the sampling parameter is set as $p =\bp$; we discuss this choice of $p$ in \cref{sec:global_tech}. Note that in this case, with high probability $G$ only makes $O(n^{1-1/2\locality^2})$ queries to the input (and so, if it exceeds the desired query complexity, it can simply reject). However, since it is information theoretically impossible to recover a $k$-bit message via $o(k)$ queries, this would imply that $n = \Omega \left( k^{1+ \frac{1}{2\locality^2 - 1}} \right)$, which yields the desired lower bound. See \cref{sec:putting-it-all-together} for a precise argument.

Thus, we are left with the task of showing that the set of queries that $G$ makes with parameter $p = \bp$ can suffice for simultaneously emulating $k$ invocations of the decoder $D$ with respect to each decoding index $i\in[k]$. To this end, we shall first need to make a simple, yet important observation regarding relaxed local decoders, which we discuss next.

\subsection{First step towards a global decoder}
\label{sec:first-step}
Recall that we assumed the relaxed local decoder $D$ is a non-adaptive algorithm with logarithmic randomness complexity, which gets query access to a string $w \in \bitset^n$. Thus, we can represent $D$ as a collection of distributions $\set{\mu_i}_{i\in[k]}$ over subsets of $[n]$ of size $\locality$, and functions $\set{f_i \colon \bitset^{\locality} \to \tritset}_{i\in[k]}$ as follows. 

For every $i \in [k]$, the distribution $\mu_i$ corresponds to the choice of local view of the relaxed decoder $D(i)$ (i.e., $D$ on decoding index $i$). The function $f_i$ is the predicate according to which $D(i)$ decides whether to decode $0$, $1$, or reject (output $\bot$) given a local view $w|_I$, where $I$ is drawn from $\mu_i$.\footnote{In fact, the decoder may rule according to a predicate that depends on the query set, but we ignore this subtlety in the high-level overview.} Note that since we assumed that $D(i)$ has logarithmic randomness complexity (i.e., $\log(n) + O(1)$, in the blocklength $n$), we have that $\mu_i$ is supported on a linear number of local views (i.e., sets of size at most $\locality$); we use this in \cref{sec:generalised-sunflower}, where we apply a combinatorial lemma on the support of $\mu_i$.

Naively, we would have liked our global decoder $G$, which queries each location with probability $p = \bp$, to emulate an invocation of the relaxed local decoder $D(i)$ by obtaining a local view of $w$ restricted to $I \sim \mu_i$. Indeed, if the distribution $\mu_i$ is ``well spread'', the probability of obtaining a local view of $D(i)$ is high. Suppose, for instance, that all of the local views are pairwise disjoint. In this case, the probability of $G$ obtaining any particular local view is at least $p^\locality = 1/n^{1/2\locality}$, and since there are $\Theta(n)$ such local views, we can expect the global decoder to obtain $\Theta(n^{1-\frac{1}{2\locality}})$ local views.

Unfortunately, if $\mu_i$ is concentrated on a relatively small number of coordinates (as is the case with all non-trivial relaxed local decoders), it is highly unlikely that the global decoder $G$ would obtain a local view of $D(i)$. For example, if $D(i)$ queries the first coordinate of $w$ with probability $1$, then we can obtain a local view of $D(i)$ with probability at most $p$, which is negligible. 

Even worse, to decode different bits, the relaxed local decoder $D$ may concentrate its queries on different locations; e.g., it could be the case that for every $i \in [k]$, the decoder $D$ would query, say, location $i$ with probability $1$. And so, there could be a large number of coordinates that are heavily queried by $D$.

At this point, the approach may appear hopeless. However, this is exactly where the relaxed decoding condition of relaxed LDCs kicks in.  Recall that the relaxed locally decoder $D$ does not err (i.e., outputs the wrong value) with high probability, as long as the codeword is not too corrupted. Thus, even if we arbitrarily guess the values of highly queried coordinates that the global decoder $G$ failed to achieve, we could still emulate an invocation of $D$ on a slightly corrupted codeword.

For example, suppose that all of the local views of the relaxed local decoder $D$ contain, say, the first coordinate, but are otherwise disjoint. Then, by the discussion above, with high probability the global decoder $G$ would obtain many partial local views, which only lack the value of the first coordinate. We can then hope to rely on the ability of the relaxed local decoder $D$ to tolerate errors, and arbitrarily fill in the value of the missing coordinate.

Namely, we could consider both possible values of the first coordinate, and observe the following. For the right ``guess'' of the value of the first coordinate, \emph{all} local views would lead $D$ to decode correctly, whereas for the incorrect guess, the \emph{majority} of local views would lead $D$ to either the correct value or $\bot$, preventing a consensus on the wrong value. (See detailed discussion of this approach in \cref{sec:global_tech}.) 

However, there are two main challenges that arise when attempting to implement the foregoing strategy in the general case; namely:
\begin{enumerate}
    \item  The combinatorial structure of the local views of a relaxed local decoder may be complex and involve many intersections; and
    \item Unlike in (non-relaxed) LDCs, a \emph{relaxed} decoder may output a reject symbol $\bot$ with high probability (possibly with probability $1$), even if only one bit of the codeword is corrupted.
\end{enumerate}

We address the first challenge in \cref{sec:generalised-sunflower}, in which we make a crucial observation about the combinatorial structure of relaxed local decoders, and address the second challenge in \cref{sec:global}, where we describe our construction of the global decoder and its analysis.

\subsection{Relaxed sunflowers}
\label{sec:generalised-sunflower}
For the next discussion, fix $i \in [k]$, and denote by $\set{L_m}_m$ the set of all local views the relaxed local decoder $D$ might query on explicit input $i$. Recall that $D(i)$ queries $L_m$ with probability $\mu_i(L_m)$.

In \cref{sec:first-step}, we observed that if the sets $\set{L_m}_m$ intersect on a single coordinate, i.e., $\cap_{m} L_m = j$ for some $j \in [n]$, and are otherwise pairwise disjoint, then the global decoder $G$ queries many partial sets $L_m \setminus \set{j}$ with high probability, where not knowing the value at coordinate $j$ still leaves us with an input within the decoding radius.

More generally, since relaxed LDC can tolerate a large (constant) fraction of errors, the foregoing argument can be extended to combinatorial \textsf{sunflowers}; that is, a collection $\sunflower 
\coloneqq \set{S_m}_m$ of subsets of $[n]$ for which there exist a ``kernel'' $K \subseteq [n]$ such that:
(1) $\cap_{m} S_m = K$, and 
(2) the ``petals'' $\petals = \set{S_m \setminus K}$ are pairwise disjoint (see \cref{fig:daisies}(a)), where the kernel is small enough such that by changing the values we assign to it we would still remain within the decoding radius.

\begin{figure}%[ht]
    \centering
    \includegraphics[scale=0.48]{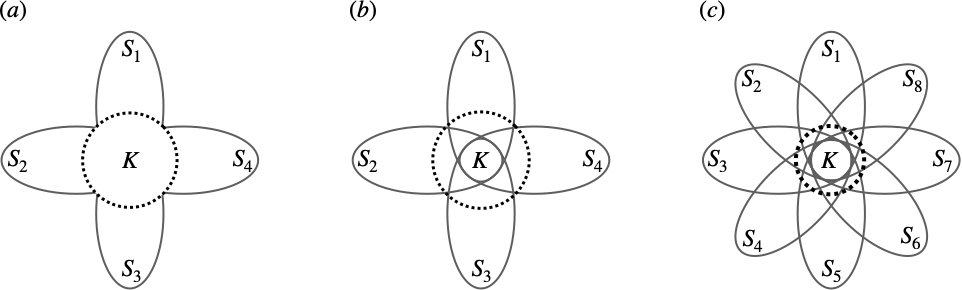}
    \caption{
    (a) \textsf{sunflower}: all sets $\set{S_m}_m$ intersect on the kernel $K$ and are otherwise pairwise disjoint;
    (b) \textsf{simple daisy (a.k.a., \simpledaisy)}: outside of the kernel $K$, all sets $\set{S_m}_m$ are pairwise disjoint; and 
    (c) $t$\textsf{-daisy}: outside of the kernel $K$, each point is covered by at most $t$ sets in $\set{S_m}_m$.
    In all figures, the dashed circle represents the kernel $K$.
    }
    \label{fig:daisies}
\end{figure}

Of course, there is no guarantee that the local views of a relaxed local decoder would be \emph{sunflowers}. While we could use \emph{sunflower lemmas} to extract sunflowers out of an arbitrary collection of subsets, we stress that the size of such sunflowers is very small (in particular, \emph{sub-linear}). Hence, we cannot simply restrict our attention to a subset of the local views that is a sunflower, as this would not preserve the soundness of the relaxed local decoder.

Nevertheless, we can use the relaxed decoding condition of $D$ even if our local views satisfy a less rigid structure than that of a sunflower. Specifically, for our argument to go through, we need the local views $\set{S_m}_m$ to only be ``mostly disjoint'' outside of an arbitrarily-structured set $K$ of small density.

Fortunately, Fischer et al. \cite{FLV15} encountered a similar combinatorial structure in the setting of \emph{property testing}, which led them to define a couple of generalisations of sunflowers, to which they refer to as ``pompoms'' and ``constellations'', and prove combinatorial lemmas regarding them. Following \cite{FLV15}, we consider a closely related relaxation of sunflowers, which we call \emph{daisies}.

Loosely speaking, a \textsf{daisy} is a sunflower in which the kernel is not necessarily the intersection of all petals, but rather a small subset such that every element outside the kernel is contained in a small number of petals.

More precisely,  a collection $\sunflower \coloneqq \set{S_m}_m$ of subsets of $[n]$ is a \daisy{t} with respect to a \emph{kernel} $K \subseteq [n]$ if for every element $j\in[n] \setminus K$, there are at most $t$ subsets $S \in \sunflower$ such that $j$ is contained in the \emph{petal} $S\setminus K$. We will refer to the special case of a \simpledaisy, wherein the petals are disjoint, as a \emph{simple daisy}. (See \cref{fig:daisies} (b) and (c).) 

Using techniques developed in \cite{FLV15}, we show a \emph{daisy lemma} (conceptually resembling a sunflower lemma) that extracts a \daisy{t} of large (constant) density, with a small kernel (i.e., such that changing its values would keep us within the decoding radius), and where $t$ is sufficiently small for obtaining petals, with high probability, using the sampling that the global decoder $G$ performs. See \cref{lem:relaxed-sunflower} for a precise statement. We remark that the daisy lemma applies to any collection of (possibly weighted) subsets, and does not rely on the fact that our collection arises from a relaxed~LDC.

In more detail, the daisy lemma extracts a \daisy{t} from the local views of each $D(i)$, which satisfies the following conditions:
\begin{itemize}[noitemsep]
    \item the density of the daisy is at least $1/\locality$;
    \item the size of the kernel $K_i$ is roughly $n^{1 - \frac{s}{\locality}}$, where $s\in[\ell]$ bounds the maximal size of the petals;
    \item the number of sets that cover each point outside of the kernel is $t = O(n^{\frac{s-1}{\locality}})$.
\end{itemize}
Next, we shall use the foregoing daisy lemma to construct the global decoder.

\subsection{The global decoder}
\label{sec:global_tech}
We are finally ready to describe our construction of the global decoder $G$ for the code $C \colon \bitset^k \to \bitset^n$, using the relaxed local decoder $D$. 

The global decoder $G$ is given query access to a string $w \in \bitset^n$, promised to satisfy $w = C(x)$, for $x \in \bitset^k$. Its goal is to fully decode $x$. To this end, $G$ starts by querying each coordinate $w_j$, for $j\in[n]$, independently with probability $p = \bp$, and tries to obtain local views of the relaxed local decoder $D(i)$ for each location $i \in [k]$, while \emph{reusing} the same samples. 

Recall that the structure of the local views of $D$ does not guarantee that any local view would be fully queried in the sampling stage of the global decoder $G$. However,  we can invoke the daisy lemma that we discussed in \cref{sec:generalised-sunflower} to extract from $\supp(\mu_i)$ a sub-collection of sets, of total density at least $1/\locality$, which is a daisy $\sunflower_i$ with kernel $K_i$, for each $i \in [k]$. Recall that the soundness error of $D$ is $1/\locality^2$, and so the density of the daisy is significantly larger than the soundness error.

Now, intuitively, since the petals of each daisy (i.e., sets $S \setminus K$ for $S \in \sunflower_i$) are ''mostly disjoint'', we can expect the global decoder $G$ to fully query a large number of petals of each daisy in $\set{\sunflower_i}_{i\in[k]}$. (See more on this at the end of \cref{sec:global_tech}.) However, it may happen that \emph{none} of the kernels $K_1, \ldots, K_k$ is queried at all (let alone fully queried) by $G$.

Thus, the main challenge is to recover the value of each $x_i$ using partial local views of $D(i)$ that do not include the value of the input $w=C(x)$ on the kernel $K_i$. For this, we shall need to rely on the properties of the relaxed local decoder, and the fact that the size of each $K_i$ is small enough such that by changing it we remain within the decoding radius of the relaxed LDC $C$.

The idea is to let the global decoder $G$ consider all possible assignments to the kernel, use each such assignment to complete the queried petals into full local views of the relaxed local decoder $D$, and rely on the properties of relaxed LDCs to identify a kernel assignment that corresponds to the decoding the correct value.

More precisely, for each $i \in[k]$, the global decoder $G$ enumerate over all possible assignments $\kappa \in \bitset^{|K_i|}$ to the kernel, and for every fully queried petal $P$ of $\sunflower_i$, considers the output of $D(i)$ on each local view $S \in \sunflower_i$ that consists of $w$ restricted to the petal $P$, with the value of $\kappa$ on the kernel $K_i$.

\label{sec:global}
\begin{figure}%[ht]
    \centering
    \includegraphics[scale=0.48]{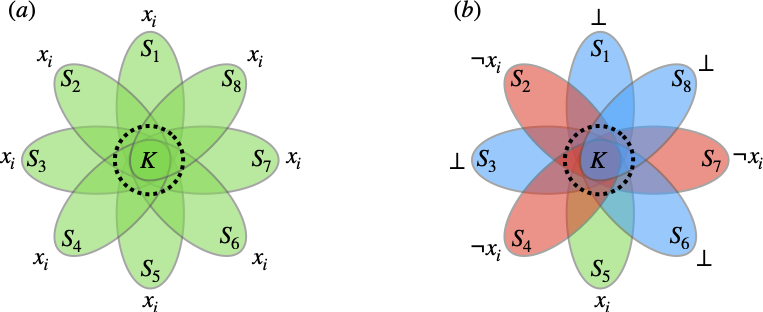}
    \caption{
    (a) corresponds to a correct ``guess'' of the kernel assignment, and (b) to a wrong one.
    Green sets correspond to local views that make the relaxed local decoder output the \emph{correct} value $x_i$.
    Red sets correspond to local views that lead to outputting the \emph{wrong} value $\neg x_i$.
    Blue sets correspond to local views that lead to reject, i.e., output $\bot$.}
    \label{fig:agreement}
\end{figure}

Now, it is crucial to make the observation that by the completeness and relaxed decoding conditions of relaxed LDCs, no kernel assignment will give rise to a majority of petals that lead to decode the wrong value, whereas there exists a least one kernel assignment that would make all petals lead to decode the correct value. More accurately, since the density of the daisy $\sunflower_i$ is at least $1/\locality$, and the soundness error is at most $1/\locality^2$, we have that:
    \begin{enumerate}
        \item Since $w$ is guaranteed to be a valid codeword $C(x)$, then for the correct assignment to the kernel $K_i$ (i.e., $\kappa \in \bitset^{|K_i|}$ such that $\kappa = w|_{K_i}$), it holds that all local views $S \in \sunflower_i$ would have made the decoder $D(i)$ output the correct value $x_i$ (see \cref{fig:agreement}(a)); and,
        \item Since changing the value of the kernel $K_i$ still leaves us within the decoding radius, then for any kernel assignment $\kappa \in \bitset^{|K_i|}$, the majority of local views $S \in \sunflower_i$ would have \emph{not} made the decoder $D(i)$ output the wrong value $\neg x_i$ (see \cref{fig:agreement}(b)).
    \end{enumerate}

The foregoing discussion naturally suggests that $G$ would decode each $x_i$ as follows: if there exists a kernel assignment that completes all fully-queried petals to local views that are in consensus on a single value $b$, and no kernel assignment leads to a consensus on $\neg b$, then output $b$. We will show that this indeed happens with $b = x_i$.

To see that, note that the first item above guarantees that there exists at least one kernel assignment for each $K_i$, which makes all petals lead to decoding the correct value $x_i$. The second item guarantees that, with high probability over petals of $\sunflower_i$ that $G$ fully queried, no kernel assignment would lead the majority of completions of fully-queried petals to local views that are consistent with the wrong value $\neg x_i$. Thus the global decoder can enumerate over all kernel assignments, and decode according to the kernel assignment that leads to a consensus.

We stress that the global decoder $G$ \emph{cannot} just guess an arbitrary value of the kernel (which would still leave us within the decoding radius). This is because it could be the case that for some kernel assignments, the majority of petals would lead to $\bot$, and petals that lead to decoding the wrong value $\neg x_i$ would be actually \emph{more common} than petals that lead to the correct value $x_i$. However, the key point is that no \emph{majority} of petals would lead to decoding the wrong value, whereas there exists a kernel assignment that would lead all petals to decode the correct value.

\subsection{Analysis of the global decoder}
\label{sec:tech-analysis}
It remains to argue that with high probability, for every $i \in[k]$ the global decoder $G$, described in \cref{sec:global_tech}, will successfully obtain fully-queried petals that would lead it to correctly decode $x_i$; that is, a set of petals such that: (1) no kernel assignment will give rise to a majority of petals that lead to decode the wrong value, and (2) there exists at least one kernel assignment that would make all petals lead to decode the correct value.

Observe that it suffices to show that for every $i \in [k]$ and kernel assignment $\kappa \in \bitset^{|K_i|}$, the global decoder $G$ only needs to obtain a single petal that leads to either outputting the correct value $x_i$ or the reject symbol $\bot$, given the kernel assignment $\kappa$; we shall refer to such petals as ``good''. To see this, note that $G$ only accepts if all petals it queried are in consensus regarding the decoding value, and so, as long there is at least one petal that corresponds to $x_i$ or $\bot$, the global decoder $G$ will not output the wrong value $\neg x_i$. On the other hand, we know that there exists a kernel assignment for which all petals lead to output the correct value $x_i$ (and neither $\bot$ nor $\neg x_i$), and so, as long as $G$ obtains one petal, it would output $x_i$.

To show that for every $i \in [k]$ and kernel assignment $\kappa \in \bitset^{|K_i|}$, the global decoder $G$ obtains a good petal as above, we first remove all petals that lead to decoding the wrong value. Recall that the density of the daisy $\sunflower_i$ we obtained from the daisy lemma is at least $1/\locality$, and the soundness error of the relaxed local decoder $D$ is at most $1/\locality^2$. Thus the total density of the good petals is at least $1/\locality - 1/\locality^2$.\footnote{In fact, it suffices have soundness error, say, $1/(10\locality)$. However, reducing the error to $1/\locality^2$ has negligible cost and it makes calculations slightly cleaner.} 

Next, we show that the density of the good petals implies that there are many of them. To this end, observe that since the density of the good local views is larger than the soundness error, then the fractional size of the set of all elements covered by the good petals must be larger than the (constant) decoding radius. This is because otherwise, replacing these elements with the values of a codeword that disagrees with $x_i$ would leave $w$ within the decoding radius, and thus break the soundness condition. Hence, the good petals cover a linear amount of coordinates, and since each petal is of constant size, we have a linear number of good petals.

Since there are many good petals in the daisy $\sunflower_i$, we can apply a lemma, which we call the \emph{simple daisy lemma} (see \cref{lem:simple-daisy}), that extracts a simple daisy (i.e., \simpledaisy{} in which the petals are pairwise disjoint) from the set of all good local views in our \daisy{t} (where recall that $t = O(n^{\frac{s-1}{\locality}})$ and $s$ bounds the maximal size of each of the petals). The resulting simple daisy has the same kernel $K_i$ of size roughly $n^{1 - \frac{s}{\locality}}$, and number of petals, each of size at most $s$, that is larger than the size of kernel $K_i$ by a multiplicative factor of at least $n^{1/\locality}$.

The petals of a simple daisy are \emph{disjoint}, and so, observe that the probability of querying all (at most $s$) elements of any petal during the sampling step is at least $p^s = 1/n^{s/2\locality^2} \geq 1/n^{1/2\locality}$. Since our simple daisy contains $d \coloneqq \Omega( n^{\frac{1}{\locality}}\cdot n^{1 - \frac{s}{\locality}})$ pairwise disjoint petals, the probability that no good petal was queried is bounded by
\begin{equation*}
    (1-p^s)^d \leq 
    \left( 1 - \frac{1}{n^\frac{1}{2\locality}}\right)^{\Omega(n^{\frac{1}{\locality}}\cdot n^{1 - \frac{s}{\locality}})} \leq
    %2^{-\Omega(n^{1 - \frac{s}{\locality}})} \leq 
    \frac{2^{-\Omega(|K_i|)}}{10k} \enspace.
\end{equation*}

Thus, for any $i \in [k],$ taking a union bound over all kernel assignments $\kappa \in \bitset^{|K_i|}$, with probability at least $9/(10k)$, for every kernel assignment we find a least one good petal, which by the discussion above implies that the global decoder $G$ successfully decodes the correct value $x_i$. Finally, taking another union bound over all decoding indices $i\in[k]$, we obtain that the foregoing holds for all $i \in [k]$ \emph{simultaneously} with probability at least $9/10$.

As discussed in \cref{sec:approach}, since it is information theoretically impossible to recover a $k$-bit message via $o(k)$ queries, and since the global decoder $G$ decoders $k$ message bits via $O(n^{1-1/2\locality^2})$ queries to the input, we deduce that $n = \Omega \left( k^{1+ \frac{1}{2\locality^2 - 1}} \right)$, which yields the desired lower bound.

\section{Preliminaries}
\label{sec:preliminaries}
We begin with standard notation:
\begin{itemize}
\item We denote the \emph{absolute distance}, over alphabet $\Sigma$, between two strings $x \in \Sigma^n$ and $y \in \Sigma^n$ by $\adist(x,y) \coloneqq \left| \left\{ x_i \neq y_i \;\colon\; i \in [n] \right\} \right|$ and their \emph{relative distance} by $\dist(x,y) \coloneqq \frac{\adist(x,y)}{n}$. If $\dist(x,y) \leq \eps$, we say that $x$ is \textsf{$\eps$-close} to $y$, and otherwise we say that $x$ is \textsf{$\eps$-far} from $y$. Similarly, we denote the \emph{absolute distance} of $x$ from a non-empty set $S \subseteq \Sigma^n$ by $\adist(x,S) \coloneqq \min_{y \in S} \adist(x,y)$ and the \emph{relative distance} of $x$ from $S$ by $\dist(x,S) \coloneqq \min_{y \in S} \dist(x,y)$. If $\dist(x,S) \leq \eps$, we say that $x$ is \textsf{$\eps$-close} to $S$, and otherwise we say that $x$ is \textsf{$\eps$-far} from $S$. We denote the projection of $x \in \Sigma^n$ on $I \subseteq [n]$ by $x|_I$.

\item We denote by $A^x(y)$ the output of algorithm $A$ given direct access to input $y$ and oracle access to string $x$. Given two interactive machines $A$ and $B$, we denote by $(A^x,B(y))(z)$ the output of $A$ when interacting with $B$, where $A$ (respectively, $B$) is given oracle access to $x$ (respectively, direct access to $y$) and both parties have direct access to $z$. 
\item Throughout this work, probabilistic expressions that involve a randomised algorithm $A$ are taken over the inner randomness of $A$ (e.g., when we write $\Pr[A^x(y) = z]$, the probability is taken over the coin-tosses of $A$).
\end{itemize}

\subsection{Coding theory}
Let $k < n$ be positive integers and let $\Gamma,\Sigma$ be alphabets. A \textdef{code} $C \colon \Gamma^k \to \Sigma^n$ is an \emph{injective} mapping from messages of length $k$ (over the alphabet $\Gamma$) to codewords of length $n$ (over the alphabet $\Sigma$). Typically it will be the case that $\Gamma=\Sigma$, in which case we simply say that the code is over the alphabet $\Sigma$. We denote by $n$ the \textdef{blocklength} of the code (which we think of as a function of $k$) and by $k/n$ the \textdef{rate} of the code. The \textdef{relative distance} of the code is the minimum, over all distinct messages $x,y \in\Gamma^k$, of $\dist(C(x),C(y))$. We shall sometimes slightly abuse notation and use $C$ to denote the set of all of its codewords $\{C(x)\}_{x\in\Gamma^k} \subset \Sigma^n$.

\subsection{Locally decodable codes}
First, we define the notion of (non-relaxed) locally decodable codes.

\begin{definition}[Locally Decodable Codes (LDCs)]
  Let $C \subseteq \Sigma^n$ be a code with relative distance $\delta_C$. We say that $\mathcal{C}$ is a \textdef{locally decodable code} if there exists a constant \textdef{decoding radius} $\decrad<\delta_C/2$ and a polynomial time
  algorithm $D$ that gets oracle access to a string $w \in \Sigma^n$ and explicit input $i \in [k]$, such that
  
  \begin{enumerate}
  \item \textsf{(Perfect) Completeness}: For any $i\in [k]$ and $w = C(x)$, where $x\in\Sigma^k$, it holds that $D^{w}(i)= x_i$.
  
  \item \textsf{Decoding}: For any $i\in [k]$ and any $w \in\Sigma^{n}$ that is $\decrad$-close to a (unique) codeword $C(x)$,
    \begin{equation*}
      \Pr[D^{w}(i) = x_i] \ge 2/3. 
    \end{equation*}
  \end{enumerate}
The \textdef{query complexity} of $D$ is the maximal number of queries that $D$ makes for any input $i$ and $w$.
\end{definition}
\noindent Note that the constant $2/3$ can be amplified as usual by repeating the process multiple times and outputting the majority symbol.

Relaxed locally decodable codes \cite{BGHSV04} are defined as follows.

\begin{definition}[RLDC]\label{def:rldc}
  A code $C\colon\Sigma^k \to \Sigma^n$ is an $\locality$-local relaxed LDC (RLDC) if there exists a constant $\decrad \in (0,\delta_C/2)$ and a randomised algorithm $D$, known as a \emph{relaxed decoder}, that on explicit input $i\in [k]$ makes $\locality$ queries to an oracle $w$ and satisfies the following conditions.
  
  \begin{enumerate}
  \item \textsf{(Perfect) Completeness}: For any $i\in [k]$ and $w = C(x)$, where $x\in\Sigma^k$, it holds that $D^{w}(i)= x_i$.
  
  \item \textsf{Relaxed Decoding}: For any $i\in [k]$ and any $w \in\Sigma^{n}$ that is $\decrad$-close to a (unique) codeword $C(x)$,
    \begin{equation*}
      \Pr[D^{w}(i) \in \{x_i,\bot\}] \ge 2/3. 
    \end{equation*}
    
    \item \textsf{Success Rate}: There exists a constant $\rho>0$ such that for any $w\in\bitset^n$ that is $\decrad$-close to a codeword $C(x)$, there exists a set $I_w\subseteq [l]$ of size at least $\rho k$ such that for every $i\in I_w$, 
    \begin{equation*}
        \Pr[D^w(i)=x_i]\geq 2/3 \enspace.
    \end{equation*}
  \end{enumerate}
\end{definition}
The \emph{randomness complexity} of a relaxed decoder is the maximal number of random coin tosses it requires to select its local view (i.e., the set of queries it makes), where the maximum is taken over the index $i\in[k]$ with respect to which it is invoked.

\section{Proof of \cref{thm:main}}
\label{sec:proof}
Let $C \colon \bitset^k \to \bitset^n$ be an $\locality$-local relaxed LDC with decoding radius $\decrad = \Omega(1)$ and $\locality = O(1)$; denote its decoder by $D$. We prove that the blocklength of $C$ must satisfy
\begin{equation*}
    n = \Omega\left(k^{1 + \frac{1}{2^{2\locality} \cdot \log(\locality)^2 - 1}}\right) \enspace.
\end{equation*}
Since $\locality = O(1)$, this means that $C$ must have nearly-linear length; that is, there exists a constant $\alpha = \alpha(\locality,\decrad)>0$ such that $n = \Omega(k^{1+\alpha})$.

To this end, we begin in \cref{sec:preprocessing} by preprocessing the relaxed local decoder $D$ and endowing it with properties that make it amenable to our techniques. Then, in \cref{sec:daisy} we present two combinatorial lemmas that would play a key role in our analysis, allowing us to structurally argue about relaxed decoders via subsets of their local views.

Next, in \cref{sec:construction} we implement the strategy presented in \cref{sec:techniques}, by using the relaxed local decoder to obtain a construction of a global decoder $G$ that receives a \emph{valid}
codeword and decodes its entire message using query complexity that is sublinear in the blocklength $n$. We analyse the global decoder in \cref{sec:analysis}, and finally, in \cref{sec:putting-it-all-together} we derive the lower bound in \cref{thm:main} from the analysis of the global decoder.

\subsection{Preprocessing}
\label{sec:preprocessing}
Our argument begins by endowing the decoder $D$ with three properties that would facilitate the analysis of our lower bound, namely: 
(1) non-adaptive queries, 
(2) reduced error probability, and
(3) logarithmic randomness complexity.

All three steps of the preprocessing step follow from straightforward adaptation of standard techniques to the setting of relaxed LDCs. We defer their proofs to \cref{apx:deferred}. We begin by transforming the decoder $D$ into a non-adaptive algorithm via a standard transformation. We note that this transformation increases the query complexity, while preserving or improving the rest of the parameters of $D$. This would later allow us to represent the behaviour of a relaxed decoder by the distribution over its local views.

\begin{claim}[Non-adaptive queries]
    \label{clm:non-adaptive}
    If there exists an (adaptive) $\locality$-local relaxed decoder for a code $C \colon \bitset^k \to \bitset^n$, then $C$ also has a \emph{non-adaptive} $2^\locality$-local relaxed decoder with the same decoding radius.
\end{claim}
See \cref{apx:non-adaptive} for the proof of \cref{clm:non-adaptive}. Denote by $D_1$ the non-adaptive relaxed decoder obtained by applying \cref{clm:non-adaptive} to the (adaptive) relaxed decoder $D$.

Next, we amplify the soundness of the decoder $D_1$, as later we shall need to invoke it multiple times and tolerate a union bound over all invocations. To this end, we use the following simple claim.

\begin{claim}[Amplification]
    \label{clm:amp}
    If there exists a non-adaptive $\locality$-local relaxed decoder for the code $C$, which errs with probability at most $1/3$, then $C$ also has an amplified non-adaptive $O(\locality \cdot \log(1/\eps))$-local relaxed decoder that errs with probability at most $\epsilon$; furthermore, the amplified relaxed decoder preserves the perfect completeness condition.
\end{claim}
See \cref{apx:amp} for the proof of \cref{clm:amp}. The transformation in \cref{clm:amp} hampers the success rate condition of a relaxed LDC, however, we stress that our argument does \emph{not} rely on the aforementioned condition. Denote by $D_2$ the amplified relaxed decoder obtained by applying \cref{clm:amp} to the relaxed decoder $D_1$ with respect to soundness error $\eps = 1/\locality^2$. 

Finally, we generalise a lemma due to Goldreich and Sheffet \cite{GS10}, which reduces the randomness complexity of query algorithms, to the setting of relaxed LDCs. This would later allow us to invoke a relaxed sunflower lemma (see \cref{sec:daisy}).

\begin{claim}[Randomness reduction]
    \label{clm:rand_reduction}
    If there exists a non-adaptive, $\locality$-local relaxed decoder, where $\locality = O(1)$, for a \emph{binary} code $C$ with constant error probability $\eps$, then $C$ also has an $O(\locality)$-local relaxed local decoder with the same parameters, except randomness complexity $\log(n) + O(1)$.
\end{claim}
See \cref{apx:rand_reduction} for the proof of \cref{clm:rand_reduction}. Denote by $D'$ the relaxed decoder obtained by applying the randomness reduction in \cref{clm:rand_reduction} to the relaxed decoder $D_2$. 

We conclude this subsection by observing that $D'$ is \emph{non-adaptive} $\locality'$-local relaxed decoder with soundness error $\eps'$ and randomness complexity $r'$, where $\locality' = O( 2^{\locality} \cdot \log(\locality) ) = O(1)$, $\eps' = 1/\locality'^2$, and $r' = \log(n) + O(1)$.

\subsection{Relaxed sunflower lemmas}
\label{sec:daisy}
As discussed in the technical overview, we shall view the non-adaptive $\locality'$-local relaxed decoder $D'$ obtained in \cref{sec:preprocessing} as a set of distributions $\set{\mu_i}_{i\in[k]}$, where each $\mu_i$ is the distribution over subsets of $[n]$ of size $\locality'$, which correspond to the local views of the relaxed decoder $D'$ on decoding index $i\in[k]$.

To argue about the probability of obtaining local views of a (non-adaptive) relaxed local decoder, we shall consider a structured subset of the local views, which satisfies a relaxed form of a combinatorial sunflower, to which we refer to as a \emph{daisy}. Loosely speaking, a \daisy{t} is a sunflower in which the kernel is not the intersection of all petals, but rather a subset such that every element outside the kernel is contained in at most $t$ petals. A formal definition follows.

\begin{definition}[daisy]
  \label{def:relaxed-sunflower}
  Suppose $U$ is a universe set and $\sunflower$ is a collection of subsets of $U$. The collection $\sunflower$ is an \emph{\daisy{t}} if there is a subset $K \subseteq U$, called \emph{kernel}, such that for every $u\in U\setminus K$, there are at most $t$ subsets $S \in \sunflower$ such that $u$ is contained in the \emph{petal} $S\setminus K$.
\end{definition}
\noindent We will refer to the special case of a \simpledaisy, wherein the petals are disjoint, as a \emph{simple daisy}.

We remark that the notion of a daisy is a generalisation of a sunflower, which provides a unified view of the notions of ``pompoms'' and ``constellations'', defined in \cite{FLV15}, without insisting on petals of equal size. 

 In the following it would be convenient to define the \emph{degree} of an element $u \in U$ in a collection $\sunflower \subseteq 2^U$ by $\deg{\sunflower}{u} = |\set{S \in \sunflower : u \in S}|$. Using this notation, a \daisy{t} $\sunflower$ satisfies that every point $u$ outside its kernel has $\deg{\sunflower}{u} \leq t$.

The main parameters of a \daisy{t} $\sunflower \subseteq 2^U$ are:
(1) the \emph{kernel size} $|K|$, 
(2) the \emph{number of sets} $|\sunflower|$, and
(3) the \emph{degree bound} $t$.
As with sunflowers, we are typically interested in finding a large daisy in a collection of subsets. More generally, since in our setting the collection of subsets would correspond to local views of a relaxed decoder, which are chosen according to some distribution, we wish to find a ``heavy'' daisy in a collection of \emph{weighted} subsets.

We will need two lemmas:
(1) a lemma that takes a weighted collection of subsets and extracts a \daisy{t}, for $t$ that is sublinear in $n$, which consists of heavy sets; and
(2) a lemma that takes a \daisy{t} and extracts a simple daisy (i.e., \simpledaisy), such that the union of its petals covers a large fraction of the domain.

Note that the former lemma yields a guarantee regarding the \emph{weight} of the sets in a daisy, which corresponds to structure representing a relaxed local decoder according to its distribution of local views, whereas the latter lemma is a purely combinatorial lemma that ``flattens out'' the weights and gives a guarantee regarding the number of coordinates that are covered by a simple daisy that is derived from a \daisy{t}.  

In the rest of this subsection, we will follow the approach in \cite{FLV15} to derive the aforementioned lemmas, which would be instrumental to our approach.

\paragraph{Extracting a heavy daisy from a weighted collection of subsets.}
The following lemma shows that a sufficiently large collection of subsets, weighted according to a distribution, contains a daisy with a small kernel and heavy petals.  

\begin{lemma}[daisy lemma]
    \label{lem:relaxed-sunflower}
    Let $\mathcal{T}$ be a collection of $c n$ subsets of $[n]$ of size $\locality$ each.
    Let $\mu$ be a distribution over $2^{[n]}$, whose support is $\mathcal{T}$.
    Then, for some $s \in [\locality]$, and $m= \max\{1,s-1\}$,
    there exists a $c n^{m/\locality}$-daisy $\sunflower \subseteq \mathcal{T}$ with a kernel of size at most $\locality \cdot n^{1-s/\locality}$ and petals of size at most $s$, such that $\mu(\sunflower) \geq 1/\locality$.
\end{lemma}
We prove \cref{lem:relaxed-sunflower} following the approach in \cite{FLV15}, and we defer this proof to \cref{apx:gen_sunflower}. 

\paragraph{Extracting a simple daisy from \daisy{t}.}
The second lemma that we shall need shows that every \daisy{t} that covers a large part of the universe set contains a simple daisy (i.e., \simpledaisy) of significant size.

\begin{lemma}[simple daisy lemma]
    \label{lem:simple-daisy}
      Let $\sunflower$ be a $t$-\emph{daisy} with kernel $K$ and petals of size at most $s$, such that $|\bigcup_{S\in \sunflower}S| = c' n$.
      Then, there exists a \emph{simple} daisy $\sunflower_0 \subseteq \sunflower$ whose kernel is $K$, such that $|\sunflower_0| \geq c' n - |K|$ if $s=1$, and otherwise $|\sunflower_0| \geq \frac{c' n - |K|}{ts^2}$.
\end{lemma}

\begin{proof}
     Initiate $\sunflower_0$ to be the empty set. Our strategy would be to iteratively add \emph{pairwise disjoint} sets to $\sunflower_0$ until it satisfies the requirements of the lemma. All the sets we shall add to $\sunflower_0$ are elements in the \daisy{t} $\sunflower$, and therefore $\sunflower_0\subseteq\sunflower$, which in turn implies that, for every $T\in \sunflower_0$, it holds that $|T\setminus K| \leq t$. In the following, denote the subset of the domain that $\sunflower$ covers by $M = \bigcup_{T\in \sunflower}T$.

    We first deal with the simple case where $s=1$; that is, all petals $\set{T \setminus K : T \in \sunflower}$ of the daisy $\sunflower$ contain just a single element. The idea is that here, for any point outside the kernel we can choose a single set that contains the point, obtaining a simple daisy.
    
    In more detail, by definition of $M$, for every element $j\in M\setminus K$, which lies outside the kernel, there exists at least one petal of $\sunflower$ that covers it, i.e., there exists $T\in \sunflower$ such that $T\setminus K = \set{j}$; choose such a set and add it to $\sunflower_0$. By our construction of $\sunflower_0$, for every distinct sets $T_1$ and $T_2$ in the daisy $\sunflower_0$, it holds that the petals $T_1\setminus K$ and $T_2\setminus K$ are disjoint. Now, Since for every point $j\in M\setminus K$ outside the kernel, there exists at least one $T\in \sunflower_0$, such that $T\setminus K = \set{j}$, we have that  $|\sunflower_0|\geq c_2n - |K|$, as required.
    
    We now proceed to the other case, where $1 < s \leq \locality$. The idea here is that, as before, for any point outside the kernel we choose a single set that contains the point, only now, we remove all sets that intersect outside the kernel with the set that we chose, and proceed this way iteratively. Details follow.

    For every $j\in M\setminus K$, there exists at least one $T\in \sunflower$ such that $j\in T\setminus K$; choose such a set and add it to $\sunflower_0$. In addition, remove $T$ from $\sunflower$ and also every set $T^*$ in $\sunflower$ such that $T\setminus K$ and $T^*\setminus K$ are not disjoint. We then reset $M$ to be the union of the sets in $\sunflower$ after the sets were removed from it, and repeat the foregoing process until $M\setminus K$ is empty. 
    
    In order to lower bound the size of $\sunflower_0$, we shall upper bound the loss in the cardinality of $M$ that happens in each iteration of removing sets from $\sunflower$. Note that the the size of the union of all sets that are removed from $M$ in a single iteration is at most their number times their size, which is bounded by $t s^2$. This is the maximum loss in the cardinality of $M$, since we are over-counting by assuming that every element of a set $T$ that is added to $\sunflower_0$ is contained in another set removed, and hence we can ignore the set $T$ during the computation. Consequently, we have that $|\sunflower_0|\geq\frac{c' n - |K|}{ts^2}$.
\end{proof}

\subsection{Construction of a global decoder}
\label{sec:construction}
Recall that in \cref{sec:preprocessing} we obtained a \emph{non-adaptive} $\locality'$-local relaxed decoder $D'$ for the code $C$ with soundness error $\eps'$ and randomness complexity $r'$, where
\begin{equation*}
    \locality' = O( 2^\locality \cdot \log(\locality)) \enspace,\enspace \eps' = 1/\locality'^2 \enspace,\enspace r' = \log(n) + O(1) \enspace.
\end{equation*}

Since $D'$ is an $\locality'$-local \emph{non-adaptive} decoder, the queries that it makes can be described by a distribution over $\locality'$-tuples of coordinates. More precisely, for every message location $i \in [k]$, there exist a predicate $f_i \colon \bitset^{\locality'} \to \tritset$ and distribution $\mu_i$ over size $\locality'$ subsets of $[k]$, such that $D'(i) = f_i(w|_I)$ for $I \sim \mu_i$.\footnote{More accurately, for any $i \in [k]$ and query set $I$ chosen by the decoder $D'(i)$, there exist a predicate $f_{i,I} \colon \bitset^{\locality'} \to \tritset$. To simplify notation, since the query set $I$ would always be clear from the context, we write $f_i$ to refer to the corresponding $f_{i,I}$.} Hereafter, we will identify the relaxed decoder $D'(i)$ with the predicate-distribution pair $(f_i, \mu_i)$.

Using the relaxed decoder $D'$, we construct a sample-based \emph{global} decoder for the code $C$, which with high probability decodes the \emph{entire} message of a \emph{perfectly valid} codeword, using $O(\sbparam)$ samples. Since it is information theoretically impossible to recover a $k$-bit message via $o(k)$ queries, this would imply that $n = \Omega(k^{\locality'^2/(2\locality'^2 - 1)})$. (See \cref{sec:putting-it-all-together} for a precise argument.) 

Loosely speaking, the global decoder works as follows. First, it samples each coordinate independently with certain probability $p$ and tries to obtain local views of the relaxed local decoder $D'(i)$ for each location $i \in [k]$, while \emph{reusing} the same samples.

Since the structure of the local views of $D'(i)$ does not guarantee that any local view would be captured in the aforementioned ``binomial sampling'' stage, the global decoder considers a subset of the local views of $D'$, which has the structure of a daisy; that is, it has a kernel of size that is smaller than the decoding radius such that outside of the kernel each point is covered by a small number of sets.

Our analysis will show that the binomial sampling stage is highly likely to yield petals of the daisy, but not its kernel, which is necessary to complete the petals into local views of the relaxed local decoder $D'$. To deal with that, the global decoder enumerates over all of the possible values of the kernel, and if one of the kernel assignments leads to a consensus of the decoding values for the corresponding local views, it outputs that value. (See \cref{sec:techniques} for a more detail high-level overview of the global decoder).

A precise description of the global decoder is given next.

\begin{construction}
\label{con:global-decoder}
\upshape
Let $\{(f_i, \mu_i)\}_{i \in [k]}$ be the predicate-distribution pairs corresponding to the relaxed decoder $D'$ obtained in \cref{sec:preprocessing}.
For every $i \in [k]$, denote by $K_i$ and $\petals_i$ the kernel and petals of the daisy $\sunflower_i$ obtained by invoking \cref{lem:relaxed-sunflower} with respect to the support of $\mu_i$.

The global decoder $G$ receives query access to a string $w \in \bitset^n$ and performs the following steps.

\begin{enumerate}
    \item \emph{Binomial sampling:} Set $p = \binomparam$, and query each coordinate $j \in[n]$ with probability $p$. Denote by $Q$ the set of all coordinates that were queried. (If the size of $Q$ exceeds the query complexity, the global decoder can simply reject.)
    \item \emph{Local view generation:} For each $i \in [k]$, let $\petals'_i \subseteq \petals_i$ be the collection of all petals that were fully queried in the binomial sampling step (i.e., $\petals'_i = \set{P \in \petals_i : P \subseteq Q}$). For every petal $P \in \petals'_i$, denote by $w|_P$ the restriction of the input $w$ to $P$.
    \item \emph{Global decoding:} for every $i \in [k]$, decode $x_i$ by performing the following steps for every assignment $\kappa \in \bitset^{|K|}$ to the kernel.
    \begin{enumerate}
        \item For every fully-queried petal $P \in \petals'$ and set $S \in \sunflower_i$ that contains $P$, let $a_{S,\kappa}$ be the assignment to $S$ whose petal assignment is $w|_P$ and kernel assignment is $\kappa|_{S\setminus P}$. Let $A$ be the set of all such assignments.
        \item If there exists $b \in \bitset$ such that $f_i(a_{S,\kappa})=b$ for every $a_{S,\kappa} \in A$, then output $b$ and proceed to decode $x_{i+1}$.
    \end{enumerate}
\end{enumerate}
\end{construction}

We proceed to analyse \cref{con:global-decoder} in \cref{sec:analysis}.

\subsection{Analysis of the construction}
\label{sec:analysis}
The following lemma shows that the construction described in \cref{sec:construction} is a (global) decoder with sublinear query complexity (in the code's blocklength) that, with high probability, decodes the entire message encoded in a perfectly valid codeword.

\begin{lemma}
\label{lem:analysis}
The algorithm $G$ defined in \cref{con:global-decoder}. Given query access to a valid codeword $w=C(x)$ for $x \in \bitset^k$, the algorithm $G$ makes $O(\sbparam)$ queries to $w$ and satisfies that $\Pr[G^w = x] \geq 2/3$. 
\end{lemma}

\begin{proof}
    Let $x \in \bitset^k$, and denote $w=C(x)$. Let $Q$ be the set of all coordinates that the global decoder $G$ queried after sampling each element in $[n]$ with probability $p = \binomparam$. Note that by standard binomial tail bounds, with probability at least $9/10$ the total query complexity of $G$ is $|G| = O(\sbparam)$, as required (otherwise $G$ can simply abort). Suppose hereafter that this is the case.
    
    Denote by $z \in \bitset^k$ the output of $G$. We show that for every $i \in [k]$, the probability that $z_i \neq x_i$ is less than $1/10k$, and thus by a union bound $\Pr[G^w = x] \geq 9/10$. Thus, the total probability of success (including obtaining the desired query complexity) is $2/3$.
    
    Fix $i \in [k]$. Recall that the global decoder $G$ decodes $x_i$ by considering a subset $\sunflower_i$ of local views of the relaxed local decoder $D'(i)$, derived from the daisy lemma (\cref{lem:relaxed-sunflower}). More precisely, during the binomial sampling stage $G$ queries the coordinates $Q$ and obtains a collection of fully queried petals of the daisy $\sunflower_i$, which we denote by $\petals'_i = \set{P \in \petals_i : P \subseteq Q}$.
    
    However, for the global decoder $G$ to rule according to the relaxed local decoder $D'(i)$, it needs not only the fully-queried petals in $\petals'_i$, but rather the complete local views of $D'(i)$ that contain these petals, i.e., the collection of subsets  $\set{S \in \sunflower_i : P \subseteq S}$. To this end $G$ needs to obtain the value of the kernel $K_i$ of $\sunflower_i$. 
    
    Since there is no guarantee that the kernel was fully queried (i.e., that $K_i \subseteq Q$), the global decoder $G$ enumerates over all possible assignments $\kappa \in \bitset^{|K_i|}$ to the kernel, and considers the output of $D'(i)$ on each local view $S \in \sunflower_i$ that consists of $w$ restricted to a petal $P \in \petals'_i$ and value of $k$.
    
    Recall that for every fully-queried petal $P \in \petals'$ and set $S \in \sunflower_i$ that contains $P$, we denote by $a_{S,\kappa}$ the assignment to $S$ whose petal assignment is $w|_P$ and kernel assignment is $\kappa|_{S\setminus P}$. Observe that, by definition (see \cref{con:global-decoder}), the global decoder $G$ outputs the correct value $x_i$ if the set of fully-queried petals $\petals'_i$ is non empty, and the following conditions hold:
        \begin{enumerate}
        \item There exists a kernel assignment $\kappa^*$ such that for \emph{every} fully-queried petal $P \in \petals'_i$ and set $S \in \sunflower_i$ that contains $P$, the relaxed local decoder $D'(i)$ outputs the correct value $x_i$ given the local view $a_{S,\kappa^*}$ (i.e. $f_i(a_{S,\kappa^*}) = x_i$).
        \item For any other kernel assignment $\kappa \in \bitset^{|K_i|}$, \emph{there exists} a fully-queried petal $P \in \petals'_i$ and set $S \in \sunflower_i$ that contains $P$, such that $D(i)$ outputs the correct value $x_i$ or aborts given the local view $a_{S,\kappa}$  (i.e. $f_i(a_{S,\kappa}) \in \set{x_i, \bot}$).
    \end{enumerate}
    Note that the first item above guarantees that at least one kernel assignment would lead the global decoder $G$ to output decode the correct value $x_i$, whereas the second item guarantees that with high probability no kernel assignment would lead $G$ to output the incorrect value $\neg x_i$.
 
    The next two claims establish that the foregoing conditions are satisfied with high probability. In the following, recall that $\locality' = O(1)$, and by \cref{lem:relaxed-sunflower}, there exists $s \in [\locality']$ and $m = \max\{1,s-1\}$ such that, denoting $t = c n^{m/\locality'}$, the collection $\sunflower_i$ is a $t$-daisy with a kernel $K_i$ of size at most $\locality' \cdot n^{1-s/\locality'}$ and $\mu_i(\sunflower_i) \geq 1/\locality'$.
    
    \begin{claim}
        \label{clm:completeness}
        There exists a kernel assignment $\kappa^* \in \bitset^{|K_i|}$ such that the assignment $a_{S,\kappa^*}$, for every $S \in \sunflower_i$ containing a queried petal $P \in \petals'_i$, satisfies $f_i(a_{S,\kappa^*}) = x_i$. Furthermore, $\petals'_i$ is non-empty with probability at least $\frac{9}{10k}$.
    \end{claim}
    
    \begin{proof}
        Set $\kappa^*$ to be the kernel assignment that coincides with $w$; that is $\kappa = w|_{K_i}$. In this case, by definition, for every $S \in \sunflower_i$ the assignment $a_{S,\kappa^*}$ equals to the local view $w|_S$ of the relaxed local decoder $D'(i)$, for a valid codeword $w = C(x)$.
        
        By the perfect completeness of $D'$, it holds that given query access to $w = C(x)$, any local view of $D'$ leads to outputting $x_i$, and thus  $f_i(a_{S,\kappa^*}) = x_i$ for all $S \in \sunflower_i$. It remains to show that with high probability there exists a least one petal of $\sunflower_i$ that was queried in the binomial sampling stage.
        
        To this end, we next argue that not only $\sunflower_i$ has large density, but that it also covers a large fraction of the domain $[n]$. Recall that the weight that the local decoder $D'(i)$ gives $\sunflower_i$ is larger than the soundness error of $D'(i)$, i.e., $\mu_i(\sunflower_i) \geq 1/\locality' \geq \eps'$. Thus, the fractional size of the set of all elements covered by $\sunflower_i$ must be larger than the decoding radius $\decrad$ (otherwise, replacing $\sunflower_i$ with the values of a codeword that disagrees with $x_i$ would leave $w$ within the decoding radius, and thus break the soundness condition), i.e., $|\cup_{S \in \sunflower_i} S| > \decrad n$.
        
        Now, we can invoke \cref{lem:simple-daisy} to ``pluck'' intersecting petals in the \daisy{t} $\sunflower_i$, for $t = c n^{m/\locality'}$, where $m= \max\{1,s-1\}$, and derive a subset that is a simple daisy (a \simpledaisy). Namely, \cref{lem:simple-daisy} implies that there exists a \emph{simple} daisy $\sunflower^*_i \subseteq \sunflower_i$ whose kernel is $K_i$ (same as $\sunflower_i$), such that
        \begin{equation*}
            |\sunflower^*_i| \geq 
            \frac{\decrad n - |K_i|}{t s^2} \geq 
            \frac{\decrad n - \locality' \cdot n^{1-s/\locality'}}{c n^{m/\locality'} s^2} =
            \Omega \left( n^{1-(m/\locality')} \right) \enspace.
        \end{equation*}
        Denote the set of petals of $\sunflower^*_i$ by $\petals^*_i$, and note that $|\petals^*_i| = (|\sunflower^*_i| - |K_i|)/s = \Omega(n^{1-(m/\locality')})$ .
        
        Since the petals of a simple daisy are \emph{disjoint}, observe that for any petal $P \in \petals^*_i$, the probability of querying all $s$ elements of $P$ during the binomial sampling step is $p^s = n^{-{s/2\locality'^2}}$. Since $\sunflower^*_i$ contains $d \coloneqq \Omega(n^{1-(m/\locality')})$ pairwise disjoint petals, the probability that no petal of $\sunflower^*_i$ was queried is
        \begin{equation*}
            \Pr[\petals^*_i = \phi] = (1-p^s)^d = 
            \left( 1 - \frac{1}{n^{{s/(2\locality'^2)}}} \right)^{\Omega(n^{1-(m/\locality')})} \leq 
            e^{\Omega\big(-n^{1 - \frac{m}{\locality'} - \frac{s}{2\locality'^2} }\big)} \leq
            \frac{1}{10k}\enspace,
        \end{equation*}
        and so $\petals'_i$ is non-empty with probability at least $9/(10k)$, concluding the proof of \cref{clm:completeness}.
    \end{proof}
    
    \begin{claim}
        \label{clm:soundness}
        For every kernel assignment $\kappa \in \bitset^{|K_i|}$, with probability at least $1 - \frac{2^{-|K_i|}}{10k}$, there exists a queried petal $P \in \petals'_i$ and $S \in \sunflower_i$ containing $P$ such that $f_i(a_{S,\kappa}) \in \set{x_i, \bot}$.
    \end{claim}
    
    \begin{proof}
        Let $\kappa \in \bitset^{|K_i|}$ be a kernel assignment. By \cref{lem:relaxed-sunflower}, the kernel $K_i$ of the \daisy{t} $\sunflower_i$ satisfies $|K_i| \leq \locality' \cdot n^{1-s/\locality'}$. In particular, note that the fractional size of the kernel is smaller than the decoding radius $\decrad$.
        
        Recall that the global decoder $G$ gets access to a perfectly valid codeword $w=C(x)$, and emulates query access to a string $z$ that agrees with $w$ outside of the kernel and with $\kappa$ inside the kernel (i.e., $z_j = w_j$ for every $j \in [n] \setminus K_i$, and $z|_{K_i}=\kappa$). Since $|K_i| \leq \decrad n$, we have that $z$ is within the decoding radius of the relaxed local decoder $D'$.
        
        By the relaxed decoding condition of $D'$, it holds that given query access to $z$, with probability at least $1-\eps'$, the local view of $z$ chosen by $D'(i)$ would lead to either outputting the correct value $x_i$ or the abort symbol $\bot$; more precisely,
        \begin{equation}
        \label{eq:weight_of_good_views}
            \Pr_{I \sim \mu_i} \left[ f_i(z|_{I}) \in \set{x_i, \bot} \right] \geq 1-\eps' = 1 - \frac{1}{\locality'^2} \enspace.
        \end{equation}
        
        Recall that the global decoder relies on a subset (a \daisy{t}) $\sunflower_i$ of the local views of $D'(i)$ to perform the decoding. We argue that $\sunflower_i$ contains a high density (according to $\mu_i$) set of local views that correspond to local views that would lead $D'(i)$ to either outputting the correct value $x_i$ or the abort symbol $\bot$. To this end, let $\good \subseteq \sunflower_i$ be the subset of all ``good''$S \in \sunflower_i$ such that $f_i(z|_{S}) \in \set{x_i, \bot}$. Note that $\good$ is also a \daisy{t} with respect to the kernel $K_i$, where $t = c n^{m/\locality'}$.
        
        By \cref{lem:relaxed-sunflower}, we have that $\mu_i(\sunflower_i) \geq 1/\locality'$. Thus,  \cref{eq:weight_of_good_views} implies that
        \begin{equation}
            \label{eq:weight_of_good_petals}
            \mu_i(\good) \geq \frac{1}{\locality'} - \frac{1}{\locality'^2} \geq \eps' \enspace.
        \end{equation}
        Similarly to the argument \cref{clm:completeness}, we observe that not only $\good$ has large density, but that it also covers a large fraction of the domain $[n]$. More accurately, by \cref{eq:weight_of_good_petals} we have that the density of $\good$ is larger than the soundness error $\eps'$, and so $|\cup_{S \in \good} S| > \decrad n$ (otherwise, replacing $\good$ with the values of a codeword that disagrees with $x_i$ would leave $w$ within the decoding radius, and thus break the 
        soundness condition). 
        
        We conclude the proof of the claim by showing that with probability at least $1 - \frac{2^{-|K_i|}}{10k}$ there exists a least one petal of $\good$ that was fully queried in the binomial sampling stage, via a similar strategy as in \cref{clm:completeness}. To this end, we invoke \cref{lem:simple-daisy} to obtain a \emph{simple} daisy $\good^* \subseteq \good$ whose kernel is $K_i$ (same as the kernel of $\good$ and $\sunflower_i$), such that $|\good^*| \geq \Omega \left( n^{1-(m/\locality')} \right)$, where $m= \max\{1,s-1\}$. Denote the set of petals of $\good^*$ by $\petals^*$, and note that
        \begin{equation*}
            |\petals^*_i| = \frac{|\good^*| - |K_i|}{s} = \Omega\left(n^{1-(m/\locality')}\right) \enspace.    
        \end{equation*}
        
        Since the petals of a simple daisy are \emph{disjoint}, observe that for any petal $P \in \petals^*_i$, the probability of querying all $s$ elements of $P$ during the binomial sampling step is $p^s = n^{-{s/2\locality'^2}}$. Since $\good^*$ contains $d \coloneqq \Omega(n^{1-(m/\locality')})$ pairwise disjoint petals and $|K_i| \leq \locality' \cdot n^{1-s/\locality'}$, the probability that no petal of $\good^*$ was queried is
        \begin{equation*}
            \Pr[\petals^*_i = \phi] =
            %(1-p^s)^d = 
            %\left( 1 - \frac{1}{n^{{s/2\locality'^2}}} \right)^{\Omega(n^{1-(m/\locality')})} = 
            e^{\Omega\big(-n^{1 - \frac{m}{\locality'} - \frac{s}{2\locality'^2} }\big)}
            %2^{-n^{1 - \frac{s}{2\locality'}}}
            \leq \frac{2^{-|K_i|}}{10k}\enspace,
        \end{equation*}
        which proves \cref{clm:soundness}.
    \end{proof}
    
    Wrapping up the argument, for any $i \in [k]$, by \cref{clm:completeness}, \emph{there exists} a kernel assignment for $K_i$ such that with probability $9/(10k)$ there is a fully queried petal $P \in \petals'_i$ that leads $D'(i)$ to output the correct value (i.e., there exists $S \in \sunflower_i$ that contains the fully queried petal $P$ such that $f_i(a_{S,\kappa^*}) = x_i$).

    Furthermore, for any $i \in [k]$, by \cref{clm:soundness} we have that \emph{for any} kernel assignment for $K_i$, with probability at least $1 - \frac{2^{-|K_i|}}{10k}$, there is a fully queried petal $P \in \petals'_i$ that leads $D'(i)$ to either output the correct value or abort (i.e., there exists $S \in \sunflower_i$ that contains the fully queried petal $P$ such that $f_i(a_{S,\kappa^*}) \in \set{x_i, \bot}$). Taking a union bound over all possible kernel assignments $\kappa \in \bitset^{|K_i|}$, we get that with probability at least $9/(10k)$, there is a set of fully queried petals $\set{P_i \in \petals'_i}_{i \in[k]}$ such that for every $i \in [k]$, the petal $P_i$ leads $D'(i)$ to either output the correct value or abort.
    
    We thus showed that for every $i \in [k]$, the probability that the global decoder $G$ fails to decode the message bit $x_i$ is at most $1/10k$. Finally, taking another union bound over all decoding indices $i\in[k]$, we obtain that the foregoing holds for all $i \in [k]$ \emph{simultaneously} with probability at least $9/10$. This concludes the proof of \cref{lem:analysis}.
\end{proof}

\subsection{Deriving the lower bound}
\label{sec:putting-it-all-together}
Recall that we have started started with a $\locality$-local relaxed LDC $C \colon \bitset^k \to \bitset^n$ with a constant decoding radius  $\decrad$ and $\locality = O(1)$, and that we wish to show that the blocklength of $C$ satisfies $n = \Omega\left(k^{1+ \alpha}\right)$, where $\alpha = \alpha(\locality,\decrad)$ is a constant.

So far, we have shown the there exists a global decoder $G$ for the code $C$, which with probability $2/3$ decodes the \emph{entire} message of a \emph{perfectly valid} codeword, using $O(\sbparam)$ samples, where $\locality' = O( 2^\locality \cdot \log(\locality))$. The following simple claim shows that decoding an entire codeword requires a least a number of queries that is linear in the dimension of the code.
\begin{claim}
    \label{clm:wrapup}
    Let $C:\bitset^k \to \bitset^n$ be a code. If there exists a randomised algorithm $\mathcal{A}$ that makes $q$ queries to a codeword $C(x)$, for some $x\in\bitset^k$, such that $\Pr \left[ \mathcal{A}^{C(x)}= x \right] \geq 2/3$, then $k = \Omega(q)$.
\end{claim}

\begin{proof}
Suppose towards contradiction that the number of queries that $\mathcal{A}$ makes is at most $k-1$. We use (the easy direction of) Yao's minimax principle to show that this implies that $\mathcal{A}$ returns the wrong answer with probability at least $1/2$, in contradiction to the claim's hypothesis. To this end, it suffices to show that there exists a distribution $\mathcal{D}$ over $n$-bit strings on which every deterministic algorithm that makes at most $k-1$ queries errs with probability at least $1/2$.

The distribution $\mathcal{D}$ is defined by simply selecting uniformly at random a message $x \in \bitset^k$ and outputting $C(x)$. Let $\mathcal{B}$ be a \emph{deterministic} algorithm that receives an input $w$ drawn from $\mathcal{D}$ and makes at most $k-1$ queries to $w$.

Let $I$ be the set of queries that the deterministic algorithm makes. Note that $I$ is deterministically fixed and $|I| \leq k-1$. After querying the indices in $I$, the algorithm $B$ can be described by a (deterministic) mapping $f$ from $2^{|I|}$ to $2^{k}$, which maps the local view $w|_I$ to a $k$-bit message. Since $|I| \leq k-1$, the range of $f$ is of size at most $2^{k-1}$, and so it contains at most half of the possible values of $x$. Thus, with probability at least $1/2$ the input $C(x)$ drawn from $\mathcal{D}$ corresponds to an $x \in \bitset^k$ that is not in the range of $f$, and hence $\mathcal{B}$ errs.
\end{proof}

Applying \cref{clm:wrapup} with respect to the global decoder $G$ implies that $n = \Omega(k^{\locality'^2/(2\locality'^2 - 1)})$, which concludes the proof of \cref{thm:main}.

\section*{Acknowledgements}
We are grateful to Oded Goldreich for numerous insightful comments and suggestions that significantly improved the exposition of this paper. We also thank Arnab Bhattacharyya and Sivakanth Gopi for a helpful discussion regarding LDC and LCC lower bounds. We thank Noga Ron-Zewi for extended discussions on constructions of relaxed LDCs.

\bibliographystyle{alpha}
\bibliography{references}

\newcommand{\etalchar}[1]{$^{#1}$}
\begin{thebibliography}{BDYW11}

\bibitem[BDSS11]{BDSS11}
Arnab Bhattacharyya, Zeev Dvir, Amir Shpilka, and Shubhangi Saraf.
\newblock Tight lower bounds for 2-query lccs over finite fields.
\newblock In {\em 2011 IEEE 52nd Annual Symposium on Foundations of Computer
  Science}, pages 638--647, 2011.

\bibitem[BDYW11]{Barak2011}
Boaz Barak, Zeev Dvir, Amir Yehudayoff, and Avi Wigderson.
\newblock Rank bounds for design matrices with applications to combinatorial
  geometry and locally correctable codes.
\newblock In {\em Proceedings of the forty-third annual ACM symposium on Theory
  of computing}, pages 519--528, 2011.

\bibitem[BGGZ18]{Blocki18}
Jeremiah Blocki, Venkata Gandikota, Elena Grigorescu, and Samson Zhou.
\newblock Relaxed locally correctable codes in computationally bounded
  channels.
\newblock In {\em 45th International Colloquium on Automata, Languages, and
  Programming (ICALP 2018)}, 2018.

\bibitem[BGH{\etalchar{+}}04]{BGHSV04}
Eli Ben{-}Sasson, Oded Goldreich, Prahladh Harsha, Madhu Sudan, and Salil~P.
  Vadhan.
\newblock Robust {PCP}s of proximity, shorter {PCP}s and applications to
  coding.
\newblock In {\em Proceedings of the 36th Annual {ACM} Symposium on Theory of
  Computing ({STOC})}, 2004.

\bibitem[BGT16]{BGT16}
Arnab Bhattacharyya, Sivakanth Gopi, and Avishay Tal.
\newblock Lower bounds for 2-query lccs over large alphabet.
\newblock {\em arXiv preprint arXiv:1611.06980}, 2016.

\bibitem[CG18]{CG18}
Cl{\'{e}}ment~L. Canonne and Tom Gur.
\newblock An adaptivity hierarchy theorem for property testing.
\newblock {\em Computational Complexity}, 27(4):671--716, 2018.

\bibitem[CGdW09]{CGW09}
Victor Chen, Elena Grigorescu, and Ronald de~Wolf.
\newblock Efficient and error-correcting data structures for membership and
  polynomial evaluation.
\newblock {\em arXiv preprint arXiv:0909.3696}, 2009.

\bibitem[DJK{\etalchar{+}}02]{Deshpande02}
Amit Deshpande, Rahul Jain, Telikepalli Kavitha, Satyanarayana~V Lokam, and
  Jaikumar Radhakrishnan.
\newblock Better lower bounds for locally decodable codes.
\newblock In {\em Proceedings 17th IEEE Annual Conference on Computational
  Complexity}, pages 184--193. IEEE, 2002.

\bibitem[DSW17]{DSW17}
Zeev Dvir, Shubhangi Saraf, and Avi Wigderson.
\newblock Superquadratic lower bound for 3-query locally correctable codes over
  the reals.
\newblock {\em Theory of Computing}, 13(1):1--36, 2017.

\bibitem[Efr12]{Efremenko12}
Klim Efremenko.
\newblock 3-query locally decodable codes of subexponential length.
\newblock {\em {SIAM} J. Comput.}, 41(6):1694--1703, 2012.

\bibitem[FLV15]{FLV15}
Eldar Fischer, Oded Lachish, and Yadu Vasudev.
\newblock Trading query complexity for sample-based testing and multi-testing
  scalability.
\newblock In {\em Proceedings of the IEEE 56th Annual Symposium on Foundations
  of Computer Science (FOCS)}, 2015.

\bibitem[GG16]{GG16a}
Oded Goldreich and Tom Gur.
\newblock Universal locally verifiable codes and 3-round interactive proofs of
  proximity for {CSP}.
\newblock {\em Electronic Colloquium on Computational Complexity {(ECCC)}},
  23:192, 2016.

\bibitem[GG18]{GG18}
Oded Goldreich and Tom Gur.
\newblock Universal locally testable codes.
\newblock {\em Chicago J. Theor. Comput. Sci.}, 2018.

\bibitem[GGK15]{GGK15}
Oded Goldreich, Tom Gur, and Ilan Komargodski.
\newblock Strong locally testable codes with relaxed local decoders.
\newblock In {\em 30th Conference on Computational Complexity, {CCC} 2015, June
  17-19, 2015, Portland, Oregon, {USA}}, 2015.

\bibitem[GKST02]{GKST02}
Oded Goldreich, Howard Karloff, Leonard~J Schulman, and Luca Trevisan.
\newblock Lower bounds for linear locally decodable codes and private
  information retrieval.
\newblock In {\em Proceedings 17th IEEE Annual Conference on Computational
  Complexity}, pages 175--183, 2002.

\bibitem[Gol04]{Goldreich04}
Oded Goldreich.
\newblock Short locally testable codes and proofs.
\newblock In {\em ECCC (later appeared in Property Testing 2010)}, 2004.

\bibitem[Gol17]{Goldreich17}
Oded Goldreich.
\newblock {\em Introduction to property testing}.
\newblock Cambridge University Press, 2017.

\bibitem[GR18]{GR18}
Tom Gur and Ron~D. Rothblum.
\newblock Non-interactive proofs of proximity.
\newblock {\em Computational Complexity}, 27(1):99--207, 2018.

\bibitem[GRR18]{GRR18}
Tom Gur, Govind Ramnarayan, and Ron~D. Rothblum.
\newblock Relaxed locally correctable codes.
\newblock In {\em 9th Innovations in Theoretical Computer Science Conference,
  {ITCS} 2018, January 11-14, 2018, Cambridge, MA, {USA}}, pages 27:1--27:11,
  2018.

\bibitem[GS06]{GS06}
Oded Goldreich and Madhu Sudan.
\newblock Locally testable codes and pcps of almost-linear length.
\newblock {\em Journal of the ACM (JACM)}, 53(4):558--655, 2006.

\bibitem[GS10]{GS10}
Oded Goldreich and Or~Sheffet.
\newblock On the randomness complexity of property testing.
\newblock {\em Computational Complexity}, 19(1), 2010.

\bibitem[HSX{\etalchar{+}}12]{Huang2012}
Cheng Huang, Huseyin Simitci, Yikang Xu, Aaron Ogus, Brad Calder, Parikshit
  Gopalan, Jin Li, and Sergey Yekhanin.
\newblock Erasure coding in windows azure storage.
\newblock In {\em Presented as part of the 2012 {USENIX} Annual Technical
  Conference}, pages 15--26, 2012.

\bibitem[KdW04]{KW04}
Iordanis Kerenidis and Ronald de~Wolf.
\newblock Exponential lower bound for 2-query locally decodable codes via a
  quantum argument.
\newblock {\em Journal of Computer and System Sciences}, 69(3):395--420, 2004.

\bibitem[KS17]{KS17}
Swastik Kopparty and Shubhangi Saraf.
\newblock Local testing and decoding of high-rate error-correcting codes.
\newblock {\em Electronic Colloquium on Computational Complexity {(ECCC)}},
  24:126, 2017.

\bibitem[KT00]{KT00}
Jonathan Katz and Luca Trevisan.
\newblock On the efficiency of local decoding procedures for error-correcting
  codes.
\newblock In {\em Proceedings of the 32nd Annual {ACM} Symposium on Theory of
  Computing ({STOC})}, 2000.

\bibitem[KV10]{KV10}
Tali Kaufman and Michael Viderman.
\newblock Locally testable vs. locally decodable codes.
\newblock In {\em Approximation, Randomization, and Combinatorial Optimization.
  Algorithms and Techniques}, pages 670--682. Springer, 2010.

\bibitem[Mei09]{Meir09}
Or~Meir.
\newblock Combinatorial construction of locally testable codes.
\newblock {\em SIAM Journal on Computing}, 39(2):491--544, 2009.

\bibitem[Oba02]{Obata02}
Kenji Obata.
\newblock Optimal lower bounds for 2-query locally decodable linear codes.
\newblock In {\em International Workshop on Randomization and Approximation
  Techniques in Computer Science}, pages 39--50. Springer, 2002.

\bibitem[Tre04]{Trevisan04}
Luca Trevisan.
\newblock Some applications of coding theory in computational complexity.
\newblock {\em Electronic Colloquium on Computational Complexity (ECCC)}, 2004.

\bibitem[Vid13]{Vid13}
Michael Viderman.
\newblock Strong ltcs with inverse poly-log rate and constant soundness.
\newblock In {\em 2013 IEEE 54th Annual Symposium on Foundations of Computer
  Science}, pages 330--339. IEEE, 2013.

\bibitem[WdW05]{WW05}
Stephanie Wehner and Ronald de~Wolf.
\newblock Improved lower bounds for locally decodable codes and private
  information retrieval.
\newblock In {\em International Colloquium on Automata, Languages, and
  Programming}, pages 1424--1436. Springer, 2005.

\bibitem[Woo12]{Woodruff12}
David~P. Woodruff.
\newblock A quadratic lower bound for three-query linear locally decodable
  codes over any field.
\newblock {\em Journal of Computer Science and Technology}, 27(4):678--686,
  2012.

\bibitem[Yek08]{Yekhanin08}
Sergey Yekhanin.
\newblock Towards 3-query locally decodable codes of subexponential length.
\newblock {\em J. {ACM}}, 55(1):1:1--1:16, 2008.

\bibitem[Yek12]{Yekhanin12}
Sergey Yekhanin.
\newblock Locally decodable codes.
\newblock {\em Foundations and Trends in Theoretical Computer Science},
  6(3):139--255, 2012.

\end{thebibliography}

\appendix
\section{Deferred proofs}
\label{apx:deferred}
We provide the proofs of lemmas and claims that were deferred from \cref{sec:proof}.

\subsection{Proof of the daisy lemma}
\label{apx:gen_sunflower}
Let $\mathcal{T}$ be a collection of $c n$ subsets of $[n]$ of size $\locality$ each. Let $\mu$ be a distribution over $2^{[n]}$, whose support is $\mathcal{T}$.
We show that for some $s \in [\locality]$, and $m= \max\{1,s-1\}$, there exists a $c n^{m/\locality}$-daisy $\sunflower \subseteq \mathcal{T}$ with a kernel of size at most $\locality \cdot n^{1-s/\locality}$ and petals of size at most $s$, such that $\mu(\sunflower) \geq 1/\locality$.
    
Our high-level strategy consists of two steps: (1) we first iteratively construct a sequence of daisies $\sunflower_1,\sunflower_2,\dots,\sunflower_\locality \in 2^{[n]}$ with kernels $K_1,K_2,\dots,K_\locality \subseteq [n]$, respectively; and (2) we then argue that there exist $s \in [\locality]$ for which the daisy $\sunflower_s$ with respect to kernel $K_s$ satisfies the desired conditions. We construct this sequence of daisies as follows.

\begin{construction}
\label{con:sequence}
Given  a collection $\mathcal{T}$ of $c n$ subsets of $[n]$ of size $\locality$ each, we construct a sequence of daisies $\sunflower_1,\sunflower_2,\dots,\sunflower_\locality \in 2^{[n]}$ with kernels $K_1,K_2,\dots,K_\locality \subseteq [n]$ as follows. 
\begin{enumerate}
    \item Set $\mathcal{T}_1 = \mathcal{T}$.
    \item Perform the following steps iteratively, for $i \in [\locality]$:
    \begin{enumerate}
        \item\label{item:core} let $K_i$ be the set of all $j\in [n]$ such that $\mathrm{deg}_{\mathcal{T}_i}(j) > cn^{i/\locality}$.
        \item\label{item:newM} let $\sunflower_i$ be the family of all subsets $T\in \mathcal{T}_i$ such that $|T\setminus K_i|\leq i$.
        \item\label{item:newF} let $\mathcal{T}_{i+1}$ be $\mathcal{T}_i \setminus \sunflower_i$.
    \end{enumerate}
\end{enumerate}
\end{construction}

We proceed to show three structural claims regarding the daisies in the sequence $\sunflower_1,\sunflower_2,\dots,\sunflower_\locality$; Namely: (1) bounding the sizes of their kernels,
(2) showing their union covers the original collection $\mathcal{T}$, and
(3) bounding the number of sets that contain each point outside of their kernels.
Subsequently, we will show that at least one of these daisies satisfy all requirements of the lemma. 

We begin with the following claim, which shows that for each $i\in [\locality]$, the daisy $\sunflower_i$ has a sufficiently small kernel.

\begin{claim}\label{clm:CoreSUB}
    For every $i\in [\locality]$, it holds that $|K_i| < \locality n^{1-i/\locality}$.
\end{claim} 
\begin{proof}
    By Item~\ref{item:newF} of \cref{con:sequence}, for every $i\in [\locality]$, it holds that $\mathcal{T}_i \subseteq \mathcal{T}_1$.
    Hence, as $|\mathcal{T}_1| = cn$, we know that $|\mathcal{T}_i| \leq cn$, for every $i\in[\locality]$.
    
Fix $i\in [\locality]$.
    By Item~\ref{item:core} of \cref{con:sequence}, 
    \begin{equation}\label{equ:degCore}
        \sum_{j\in K_i} \mathrm{deg}_{\mathcal{T}_i}(j) > |K_i|cn^{i/\locality} 
    \end{equation}
    Since all the subsets of $[n]$ have cardinality $\locality$, 
    \begin{equation}\label{equ:sumDeg}
        \sum_{j\in K_i} \mathrm{deg}_{\mathcal{T}_i}(j) =  \locality|\mathcal{T}_i| \leq \locality c n
    \end{equation}
The claim follows from \cref{equ:degCore,equ:sumDeg}.
\end{proof} 

The next claim shows that the union of all subsets in all the daisies $\sunflower_1,\sunflower_2,\dots,\sunflower_\locality$ covers the original collection of subsets $\mathcal{T}$ (equivalently, $\mathcal{T}_1$).
 
 \begin{claim}\label{clm:partition}
     $\bigcup_{i\in[\locality]}\sunflower_i = \mathcal{T}_1$.    
 \end{claim}
 \begin{proof}
    Note that for $i = \locality$, the condition in Item~\ref{item:newM} of \cref{con:sequence} lets $\sunflower_\locality$ be the family of all subsets $T\in \mathcal{T}_\locality$ such that $|T\setminus K_\locality| \leq \locality$.
    The cardinality of each set in $\mathcal{T}_\locality$ is $\locality$, and hence, trivially, $\sunflower_\locality = \mathcal{T}_\locality$, which in turn means that every set in $\mathcal{T}_1$ is in one of the families in $\{\sunflower_1,\sunflower_2,\dots,\sunflower_\locality\}$ and the claim follows.
 \end{proof}
 
 Next, we show a claim which shows that for each $i\in [\locality]$, every point outside of the kernel of $\sunflower_i$ is incident in only a small number of sets of $\sunflower_i$.
 
 \begin{claim}\label{clm:External}
    For every $i\in [\locality]$, and $j\in [n]\setminus K_i$, $\mathrm{deg}_{\sunflower_i}(j) \leq cn^{\frac{\max\{1,(i-1)\}}{\locality}}$.
 \end{claim}
  \begin{proof}
  For $i=1$, the claim follows directly from Item~\ref{item:core} of \cref{con:sequence}.
  
   Suppose towards contradiction that there exists $i\in \{2,3,\dots,\locality\}$ and $j\in [n]\setminus K_i$, such that $\mathrm{deg}_{\sunflower_i}(j) > cn^{(i-1)/\locality}$. Let $T$ be a set in $\sunflower_i$ such that $j\in T$. Let $T'$ be the subset of $T$ that consists of $T\cap K_i$ and every index $h\in T\setminus K_i$ such that $\mathrm{deg}_{\sunflower_i}(h) > cn^{(i-1)/\locality}$.
   
   The size of $T'$ is at least $|T\cap K_i| + 1$, because we assumed $T'$ has an element in $T\setminus K_i$. Hence, $|T'|\geq \locality-(i-1)$.
   We next show that
   \begin{equation*}
       \mathrm{deg}_{\sunflower_{i-1}}(h) > cn^{\frac{(i-1)}{\locality}} \enspace,
   \end{equation*}
   for every $h\in T'$. This implies that $T'\subseteq K_{i-1}$ and in turn that $|T\setminus K_{i-1}| = \locality - |T'| \leq i-1$.
   Thus, $T \in \sunflower_{i-1}$ and consequently, by Items~\ref{item:newF} and~\ref{item:newM} of \cref{con:sequence}, $T$ is not in $\sunflower_i$, in contradiction to our initial assumption.
   
   By Item~\ref{item:core} of \cref{con:sequence}, for every $h\in T' \cap K_i$,
   it holds that $\mathrm{deg}_{\mathcal{T}_i}(h) > cn^{i/\locality}$, which in turn implies that 
   \begin{equation*}
       \mathrm{deg}_{\mathcal{T}_{i-1}}(h) > cn^{i/\locality}> cn^{(i-1)/\locality} \enspace,
   \end{equation*}
   because $\mathcal{T}_i\subseteq \mathcal{T}_{i-1}$, by Item~\ref{item:newF} of \cref{con:sequence}.
   
   By item~\ref{item:newM}, $\sunflower_i \subseteq \mathcal{T}_i$, and we already deduced that $\mathcal{T}_i\subseteq \mathcal{T}_{i-1}$. Thus, $\sunflower_{i}\subseteq \mathcal{T}_{i-1}$ and therefore, for every $h\in T'\setminus K_i$ if $\mathrm{deg}_{\sunflower_i}(h) > cn^{(i-1)/\locality}$, then also $\mathrm{deg}_{\mathcal{T}_{i-1}}(h) > cn^{(i-1)/\locality}$, and the claim follows.
 \end{proof}
 
 Finally, we rely on \cref{clm:CoreSUB}, \cref{clm:partition}, and \cref{clm:External} to show that at least one of the daisies in in the sequence $\sunflower_1,\sunflower_2,\dots,\sunflower_\locality$ satisfies all of the conditions of the lemma.

 \begin{claim}
    For some $s \in [\locality]$, $\sunflower_s$ is a $cn^{\max\{1,(s-1)\}/\locality}$-daisy with a kernel $K$, such that  $\mu(\sunflower_s)\geq 1/\locality$, $|K| <\locality \cdot n^{1-s/\locality}$, and $|T\setminus K|\leq s$, for every $T\in \sunflower_s$.
 \end{claim}
 
 \begin{proof}
   By Claim~\ref{clm:partition}, the sequence of daisies $\sunflower_1,\sunflower_2,\dots,\sunflower_\locality$ covers $\mathcal{T}_1$, i.e., $\bigcup_{j\in [\locality]}\sunflower_j = \mathcal{T}_1$. Therefore, there exists $s \in [\locality]$ such that $\mu(\sunflower_s)\geq 1/\locality$.
    We take $K$ to be the set $K_s$.
    According to the construction of $K_s$, for every $T\in \sunflower_s$ it holds that $|T\setminus K_s|\leq s$, and
     by Claim~\ref{clm:CoreSUB}, we have that $|K_s| < \locality n^{1-s/\locality}$. Finally, by Claim~\ref{clm:External}, $\mathrm{deg}_{\sunflower_s}(j) \leq cn^{\max\{1,(s-1)\}/\locality}$, for every $j\in [n]\setminus K_s$.
 \end{proof}
 This concludes the proof of \cref{lem:relaxed-sunflower}.
 
\subsection{Proof of \cref{clm:non-adaptive}}
\label{apx:non-adaptive}
Let $D$ be an adaptive $\locality$-local relaxed decoder for the code $C \colon \bitset^k \to \bitset^n$. We show that $C$ also has a \emph{non-adaptive} $2^\locality$-local relaxed decoder with the same decoding radius.

Fix $i \in [k]$. Note that the adaptive decoder $D$ can be viewed as a distribution over binary decision trees; that is, $D(i)$ first tosses coins to obtain a random string $\rho$ which  determines the binary decision tree that $D$ uses deterministically to determine the output $D(i)$ given query access to input $w$.

Recall that every non-leaf vertex of a decision tree is labelled by an index in $[n]$, and one of the edges leaving it towards a child is labelled $0$ and the other is labelled $1$. The relaxed local decoder $D$ uses a decision tree, which it selects at random, by starting from the root of the tree, reading its label $i$ and querying $w_i$. It then proceeds to the child of the root with the edge corresponding to the value of $x_i$. The child is treated in the same manner as the root if it is an internal vertex, and if it is a leaf, its label takes value in $\{0,1,\bot\}$, which is the output of $D(i)$.

Let $D'$ be an algorithm that operates as follows. Given explicit input $i \in [k]$, it first tosses coins, exactly like $D(i)$, to obtain a random string $\rho$ which  determines the binary decision tree that $D(i)$ uses. Then, it queries \emph{all} of the indices labelling the vertices of the decision tree (corresponding to all possible queries that $D(i)$ might have been given any possible input $w$). Finally, $D'$ uses the query values to compute the answer $D(i)$ would have returned and returns it. Note that there are $2^\locality$ such labels.

The choice of queries of $D'$ depends only on the decision tree chosen, hence it is non-adaptive. The query complexity of $D'$ is $2^\locality$ and it returns the exact same answer as $D$ would on the same input and random coin tosses. This concludes the proof of \cref{clm:non-adaptive}.

\subsection{Proof of \cref{clm:amp}}
\label{apx:amp}
Let $D$ be an $\locality$-local relaxed decoder for the code $C \colon \bitset^k \to \bitset^n$, which errs with probability at most $1/3$. We use $D$ to construct an $O(\locality \cdot \log(1/\eps))$-local relaxed decoder $D'$ that errs with probability at most $\epsilon$ and has perfect completeness.

On input $i\in [k]$ and query access to a string $w \in \bitset^n$, the relaxed local decoder $D'$ operates as follows:
\begin{enumerate}
    \item invoke $\log(1/\eps)$ parallel executions of $D$ with the same input parameters;
    \item when the invoked executions attempt to query $w$, the relaxed local decoder $D'$ collects all of the queries and asks them simultaneously, returning the values of the queries to the executions that requested them;
    \item Finally, $D'$ collects all the outputs from the executions and returns $b\in \{0,1\}$ if all the executions returned $b$, and otherwise returns $\bot$.
\end{enumerate}
Note that all the queries are used in a non-adaptive manner, and that there are at most $\locality \cdot \log(1/\eps)$ of them.

By the completeness of the original relaxed local decoder $D$, if $w$ is a valid codeword, then $b = x_i$ with probability $1$. Hence, $D'$ has perfect completeness. On the other hand, if $w$ is $\decrad$-close to $C$, then by the relaxed decoder condition of relaxed LDC, it holds that $x \not\in \{x_i,\bot\}$ with probability less than $(1/3)^{\log(1/\eps)} < \eps$. Thus, $D'$ has the desired amplified soundness. This concludes the proof of \cref{clm:amp}.

\subsection{Proof of \cref{clm:rand_reduction}}
\label{apx:rand_reduction}
Let $D$ be a $\locality$-local relaxed decoder for a \emph{binary} code $C$ with $\locality = O(1)$, randomness complexity $r$, and constant error probability $\eps$. We use $D$ to construct a relaxed decoder $D'$ with the same parameters as $D$, except it has locality $O(\locality)$ and randomness complexity $\log(n) + O(1)$.

Fix $i \in [k]$. Consider a $2^r \times 2^{n}$ matrix where the rows correspond to all possible random strings $\gamma$ used by the relaxed local decoder and the columns correspond to all inputs $w \in \bitset^n$ that are within the decoding radius of $C$. The entry $(\gamma,w)$ of the matrix corresponds to the output of $D^w(i ; \gamma)$, that is, the output of the relaxed local decoder when given query access to $w$ and random coins $\gamma$.

Note that for every codeword $w = C(x)$, by the perfect completeness of
$D(i)$ the value of each entry in a $w$ column equals the correct message value $x_i$. By the relaxed decoding condition of $D(i)$, for each $w$ column that is within the decoding radius of $C$, at least $1-\eps$ fraction of the entries are in $\set{x_i,\bot}$. 
  
We show that there exists a multi-set, $S$, of size $O(n)$ of
the rows such that the every column $w$ restricted to $S$ has at most $O(\eps)$ fraction of entries taking the wrong value $\neg x_i$. Thus, we obtain a relaxed local decoder $D'$ that uses only $\log_2 |S| = \log{n} + O(1)$ random coins, by simply running the original decoder $D$ but with respect to random coins selected uniformly from $S$ (rather than from $\bitset^r$). To obtain soundness error $\eps$ we use $O(1)$ parallel repetitions.

We use the probabilistic method to show the existence of a small multi-set $S$ as above. Consider a multi-set $S$ of the rows, of size $t$, chosen uniformly at random and fix input $w$. By the Chernoff bound, with probability $2^{-\Omega(t)}$ over the choice of $S$, at most $O(\eps)$ fraction of entries of $w$ restricted to $S$ take the wrong value $\neg x_i$. Thus, by setting $t = \log(2^n) + O(1)$ and applying the union bound, we obtain that there exists a multi-set $S$ as desired. Since the new relaxed local decoder selects at random from $S$, it can be implemented using $\log_2{t} +O(1)$ random coins. This concludes the proof of \cref{clm:rand_reduction}.

\end{document}